\documentclass[11pt]{article}
\usepackage[letterpaper,hmargin=1in,vmargin=1in]{geometry}
\usepackage{graphicx}                   % graphics (includegraphics)
\usepackage{amsmath}                    % AMS math
\usepackage{amssymb}                    % special AMS math symbols
\usepackage{amsfonts}                   % more special AMS math symbols
\usepackage{url}                        % \url{...} formatting URLs
\usepackage{color}                      % colored text and backgrounds
\usepackage{cite}                       % sort bibliographical entries
\usepackage{hyperref}                   % add hyper-references
\hypersetup{colorlinks=true,linkcolor=[rgb]{0,0,0},citecolor=[rgb]{0,0,0},urlcolor=[rgb]{0,0,0}}
\usepackage{amsthm}                     % AMS theorem/proof
\usepackage{thmtools}                    
\usepackage{thm-restate}
\usepackage{mathrsfs}                   % mathscr script fonts
\usepackage{xparse}
\usepackage{xcolor}

\declaretheorem[name=Theorem]{theorem} % LiPiCS overrides
\declaretheorem[name=Lemma,numberwithin=section]{lemma}
\declaretheorem[name=Corollary,sibling=lemma]{corollary}

\newcommand{\ang}[1]{\left\langle #1 \right\rangle}
\newcommand{\RE}{\mathbb{R}}            % real space
\newcommand{\eps}{\varepsilon}          % my preferred epsilon
\newcommand{\ST}{\,:\,}                 % { x \ST y }

\newcommand{\inv}[1]{\frac{1}{#1}}
\newcommand{\inner}[2]{\ang{#1,#2}} % inner product

\newcommand{\polar}[1]{#1^{\circ}}
\newcommand{\bd}{\partial}
\newcommand{\etal}{\textit{et al.}}
\newcommand{\SP}{\kern+1pt}
\newcommand{\orthproj}[2]{{#1}\mathbin{\big|}{#2}} % orthogonal projection
\newcommand{\orthprojsub}[2]{{#1}\mathbin{|}{#2}} % orthogonal projection
\DeclareMathOperator{\interior}{int}
\DeclareMathOperator{\proj}{proj}
\DeclareMathOperator{\dproj}{dproj}
\DeclareMathOperator{\conv}{conv}
\DeclareMathOperator{\cone}{cone}

\DeclareMathOperator{\vol}{vol}
\DeclareMathOperator{\area}{area}
\DeclareMathOperator{\dist}{dist}
\DeclareMathOperator{\finsler}{Finsler}

\NewDocumentCommand{\volHilb}{O{}}{\vol^{H}_{#1}} % Holmes-Thompson Hilbert volume
\NewDocumentCommand{\volFunk}{O{}}{\vol^{F}_{#1}} % Holmes-Thompson Funk volume
\NewDocumentCommand{\volX}{O{}}{\vol_{#1}} % generic volume
\NewDocumentCommand{\areaHilb}{O{}}{\area^{H}_{#1}} % Hilbert surface area
\NewDocumentCommand{\areaFunk}{O{}}{\area^{F}_{#1}} % Funk surface area
\NewDocumentCommand{\areaMink}{O{}}{\area^{M}_{#1}} % Minkowski surface area
\NewDocumentCommand{\areaX}{O{}}{\area_{#1}} % generic area
\NewDocumentCommand{\distHilb}{O{}}{\dist^{H}_{#1}} % Hilbert distance
\NewDocumentCommand{\distFunk}{O{}}{\dist^{F}_{#1}} % Funk distance
\NewDocumentCommand{\distX}{O{}}{\dist_{#1}} % generic distance
\NewDocumentCommand{\finsHilb}{O{}}{\finsler^{H}_{#1}} % Finsler Hilbert distance
\NewDocumentCommand{\finsFunk}{O{}}{\finsler^{F}_{#1}} % Finsler Funk distance
\NewDocumentCommand{\polarIn}{O{}m}{{#2}^{\circ_{#1}}} % polar in a specified space
\NewDocumentCommand{\ballX}{O{}}{B_{#1}} % generic ball
\NewDocumentCommand{\ballHilb}{O{}}{B^{H}_{#1}} % ball in Hilbert
\NewDocumentCommand{\ballFunk}{O{}}{B^{F}_{#1}} % ball in Funk
\NewDocumentCommand{\cauchyHilb}{O{}}{\proj^{H}_{#1}} % Holmes-Thompson projection Hilbert
\NewDocumentCommand{\cauchyFunk}{O{}}{\proj^{F}_{#1}} % Holmes-Thompson projection Funk
\NewDocumentCommand{\dCauchyFunk}{O{}}{\dproj^{F}_{#1}} % directional Holmes-Thompson projection

\title{Cauchy's Surface Area Formula in the Funk Geometry}

\author{%
	Sunil Arya\thanks{Research supported by the Research Grants Council of Hong Kong, China under project number 16214721.}\\
		Department of Computer Science and Engineering \\
		Hong Kong University of Science and Technology \\
		Clear Water Bay, Kowloon, Hong Kong\\
		arya@cse.ust.hk \\
		\and
	David M. Mount\\
		Department of Computer Science and \\
		Institute for Advanced Computer Studies \\
		University of Maryland \\
		College Park, Maryland 20742 \\
		mount@umd.edu
}

\begin{document}
\date{~}
\maketitle

%-----------------------------------------------------------------------
\begin{abstract}
Cauchy's surface area formula expresses the surface area of a convex body as the average area of its orthogonal projections over all directions. While this tool is fundamental in Euclidean geometry, with applications ranging from geometric tomography to approximation theory, extensions to non-Euclidean settings remain less explored. In this paper, we establish an analog of Cauchy's formula for the Funk geometry induced by a convex body $K$ in $\mathbb{R}^d$, for the Holmes--Thompson surface area. The formula is based on central projections to boundary points of $K$. We show that when $K$ is a convex polytope, the formula reduces to a weighted sum of contributions associated with the vertices of $K$. Finally, as a consequence of our analysis, we derive a generalization of Crofton's formula for surface areas in the Funk geometry. By viewing Euclidean, Minkowski, Hilbert, and hyperbolic geometries as limiting or special cases of the Funk setting, our results provide a unified framework for these classical surface area formulas.
\end{abstract}
%-----------------------------------------------------------------------


%=======================================================================
\section{Introduction} \label{s:intro}
%=======================================================================

Cauchy's surface area formula is a classical identity in Euclidean convex geometry. It states that the surface area of a convex body is a constant multiple of the average $(d-1)$-dimensional volume of its orthogonal projections, or ``shadows'', over all directions. To state the result formally, assume that $d \geq 2$. Let $S^{d-1}$ denote the Euclidean unit sphere, let $\sigma$ denote its rotation-invariant surface measure, let $\lambda_k$ denote $k$-dimensional Hausdorff measure, and let $\omega_d$ denote the volume of the $d$-dimensional Euclidean unit ball. For a convex body $K \subset \RE^d$ and $u \in S^{d-1}$, let $\big(\orthproj{K}{u^\perp}\big)$ denote the orthogonal projection of $K$ onto the hyperplane $u^\perp$. Cauchy's formula asserts that 
\begin{equation} \label{eq:cauchy-euclidean}
    \lambda_{d-1}(\partial K)
     ~=~
    \frac{1}{\omega_{d-1}} \int_{S^{d-1}} \lambda_{d-1} \big(\orthproj{K}{u^{\perp}}\big) \, d\sigma(u).
\end{equation}
The related \emph{Crofton formula} expresses the surface area of $K$ in terms of the average number of intersections of $\partial K$ with lines, taken with respect to a suitable measure on the space of lines.

Cauchy's formula is computationally fundamental, as it underpins algorithms in applications where estimating global geometric properties from lower-dimensional measurements is a recurring task, including geometric tomography~\cite{Gar06}, stereology~\cite{BaJ05}, and surface area estimation for digitized 3D objects~\cite{LYZ10, LWM03}. Integral formulas of this form are the essential first steps in efficient sampling processes, such as computing $\eps$-nets \cite{AHV04, Cla06}, since they implicitly define the probability distribution upon which to base the sampling.

Modern applications have increasingly explored non-Euclidean spaces, such as the \emph{Hilbert geometry} and its close relative, the \emph{Funk geometry}. (See Section~\ref{s:prelim} for definitions.) These geometries arise naturally whenever the domain of interest is a convex set. Examples include the analysis of discrete probability distributions in machine learning, where the domain is the probability simplex \cite{BNN10, NiS19}, analysis of networks through hyperbolic geometry~\cite{KPK10}, analysis of positive definite matrices in deep learning~\cite{ISY11, LDL23}, and lattice-based cryptosystems~\cite{AFM24}. Hilbert is more widely studied than Funk, but Faifman argues that for many applications, including computing volumes and areas, the Funk geometry is more natural and yields cleaner results~\cite{Fai24}. 

While the metric properties of these spaces are well-understood in differential geometry~\cite{Fai24, PaT09, PaT14}, computational aspects of these geometries are only beginning to emerge. Recent work has employed the Hilbert geometry for polytope approximation~\cite{AAFM22, AFM24}, clustering~\cite{NiS19}, and approximate membership queries~\cite{AbM24}. Despite this growing algorithmic interest, the integral geometry of these spaces lacks the computational primitives required for efficient implementations. A major barrier is that the standard definition of surface area (the Holmes-Thompson area~\cite{HoT79}) relies on symplectic forms or global integrals involving the polar body, making direct numerical evaluation prohibitively complex for high-dimensional or large-scale settings.

In this paper, we bridge this gap by establishing explicit Cauchy and Crofton formulas for the Funk geometry. Our formulas transform the abstract definitions of Finsler area into simple, concrete geometric quantities, which are readily accessible to discrete computation. Beyond the specific interest in Funk geometry, our results reveal it as a unifying bridge connecting diverse geometric settings. 

Our main result is a Funk analog of Cauchy's formula. Let $K$ be a convex body in $\RE^d$, and let $G$ be a convex set contained in the interior of $K$. Our objective is to compute the Holmes--Thompson surface area of $G$ in the Funk geometry induced by $K$, which we denote by $\areaFunk[K](G)$. For simplicity, let us assume for now that $K$'s boundary is strictly convex, implying that for each direction $u \in S^{d-1}$, there is a unique boundary point $v_K(u)$ whose supporting hyperplane is orthogonal to $u$ (see Figure~\ref{f:funk-shadows-1}(a)). Define the \emph{central shadow} $S_K(G, u)$ as the slice of the cone subtended by $G$ at $v_K(u)$, orthogonal to the direction $u$:
\[
    S_K(G,u) ~ = ~ -u^* \cap \cone(G - v_K(u)),
\]
where $-u^*$ is the hyperplane tangent to $S^{d-1}$ at $-u$ and $\cone(G - v_K(u))$ is the translation of the subtended cone to the origin (see Figure~\ref{f:funk-shadows-1}(b)). The measure of this section captures the ``visual size'' of $G$ as seen from the boundary point $v_K(u)$. Just as the Euclidean formula averages orthogonal shadows, our Funk formula averages these central shadows. We show that $G$'s surface area is proportional to the average area of these shadows, taken over all directions $u$.

%-----------------------------------------------------------------------
\begin{figure}[htbp]
  \centerline{\includegraphics[scale=0.40,page=1]{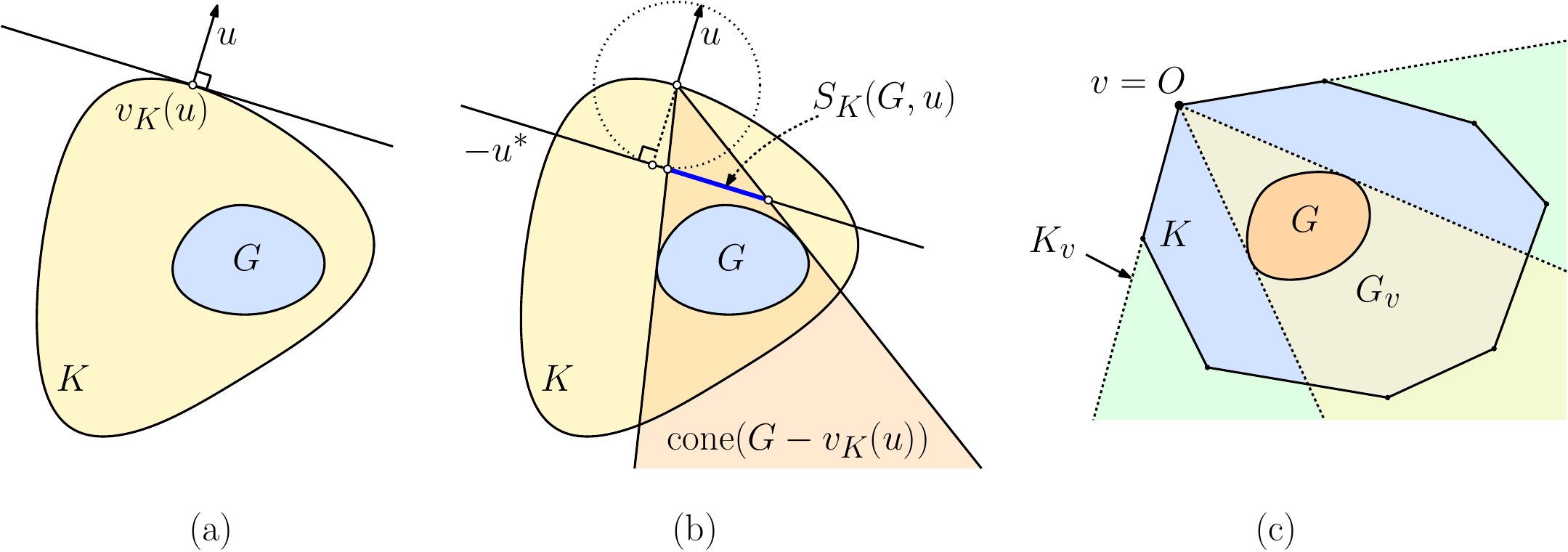}}
  \caption{(a) The boundary point $v_K(u)$, (b) the central shadow $S_K(G,u)$ (translated so that $v_K(u)$ coincides with the origin), and (c) vertex decomposition.}
  \label{f:funk-shadows-1}
\end{figure}
%-----------------------------------------------------------------------

%-----------------------------------------------------------------------
\begin{theorem} \label{thm:cauchy-funk-gen}
Let $G$ and $K$ be two convex bodies in $\RE^d$ such that $G \subset \interior(K)$. Then
\[
    \areaFunk[K](G) 
        ~ = ~ \frac{1}{\omega_{d-1}} \int_{S^{d-1}} \lambda_{d-1}(S_K(G,u)) \, d\sigma(u).
\]
\end{theorem}
%-----------------------------------------------------------------------

This provides a ``tomographic interpretation'' in the sense that the intrinsic Funk surface area can be recovered solely from these central slices, which capture the visibility of $G$ from the boundary of $K$. It is interesting to note that, while the Holmes--Thompson surface area of $G$ (presented in Section~\ref{s:ht-vol-area}) is defined in terms of the areas of the polars of Finsler balls in the Funk geometry induced by $K$, the formula of Theorem~\ref{thm:cauchy-funk-gen} involves only simple Euclidean quantities and the standard Lebesgue measure.

Our formulation connects the Funk geometry to several classical settings. When $K$ is expanded to infinity, the central shadows become parallel projections, and our result converges to a Cauchy-type formula for the surface area in Minkowski spaces (Lemma~\ref{lem:minkowski-cauchy}). In the special case where $K$ is a Euclidean ball, this limit yields the classical Euclidean Cauchy formula. Furthermore, for a finite Euclidean ball $K$, our formula yields the exact surface area in the Beltrami-Klein model of hyperbolic geometry~\cite{CFK97}, generalizing planar observations by Alexander \etal~\cite{ABF05}. Finally, for arbitrary convex bodies, our Funk surface area serves as a constant-factor approximation for the Hilbert surface area, with factors depending only on the dimension $d$ (Corollary~\ref{cor:hilbert-exact}).

Crucially for algorithmic applications, we derive a \emph{discrete surface area formula} when $K$ is a polytope. We show that the Funk surface area of a body $G$ nested within $K$ can be decomposed into a sum of local terms associated with the \emph{vertices of $K$}. Each vertex $v$ of $K$ is naturally associated with a pointed cone $K_v = \cone(K - v)$, which is the cone subtended by $K$ after translating $v$ to the origin (see Figure~\ref{f:funk-shadows-1}(c)). The body $G$ is similarly associated with a cone $G_v = \cone(G - v)$. In Section~\ref{s:cones}, we define a simple shadow-based notion of the Funk volume of $G_v$ with respect to $K_v$, which we denote by $\volFunk[K_v](G_v)$. We show that the total surface area of $G$ can be expressed as the sum of these cone-based Funk volumes over the vertices of $K$.

%-----------------------------------------------------------------------
\begin{theorem}[Vertex Decomposition for Polytopes] \label{thm:cauchy-funk-polytope}
Let $K$ be a convex polytope in $\RE^d$, and let $G \subset \interior(K)$ be a convex body. For each vertex $v \in V(K)$, let $K_v = \cone(K - v)$ and $G_v = \cone(G - v)$. Then
\[
    \areaFunk[K](G)
        ~ = ~ \sum_{v \in V(K)} \volFunk[K_v](G_v).
\]
\end{theorem}
%-----------------------------------------------------------------------

While the cone-based Funk volume defined in Section~\ref{s:cones} involves integration, it can be readily estimated through random sampling. Furthermore, this decomposition scales linearly with the vertex count of $K$ (see the discussion following the proof of Theorem~\ref{thm:cauchy-funk-gen} in Section~\ref{s:cauchy-gen}). This is in contrast to prior Crofton measures for Hilbert geometries, which involve quadratic enumerations over face pairs~\cite{Sch06a}.

Finally, we establish a \emph{Funk--Crofton formula} (Theorem~\ref{thm:crofton-unoriented}), expressing the Funk surface area as the $K$-dependent measure of the set of line parameters whose associated lines intersect the body. This result provides the basis for Monte-Carlo estimation of surface area, similar to methods used in Euclidean geometry~\cite{ACF22, LWM03}. That is, one can approximate the Funk surface area simply by sampling from the associated parameter distribution and counting intersections of the corresponding lines with the body, bypassing the need for complex analytical evaluation. 

%=======================================================================
\subsection{Related Work}
%=======================================================================

The study of integral geometry in projective Finsler spaces has a rich history in differential geometry~\cite{AlF98, Sch01, Sch06b}. Schneider~\cite{Sch01} established the existence of Crofton measures for general projective Finsler spaces using the dual Holmes--Thompson volume. For the specific case of \emph{polytopal Hilbert geometries}, he derived an explicit Crofton measure involving a combinatorial decomposition over pairs of complementary faces of the polytope~\cite{Sch06a}. Our work differs by focusing on Funk geometry and providing a decomposition based on the \emph{vertices} of the ambient polytope, which, as noted above, offers significant computational advantages.

In the planar setting, Alexander~\cite{Ale78} and later Alexander, Berg, and Foote~\cite{ABF05} derived elegant Cauchy-type perimeter formulas for the Hilbert geometry using trigonometric integrals. Note that in this case, Hilbert and Funk surface areas coincide. Like ours, their formulas are explicit and admit a shadow interpretation in the Beltrami-Klein disk model of hyperbolic geometry. However, their techniques are specialized to the plane and do not seem to extend to arbitrary convex domains in higher dimensions. In contrast, our approach unifies these perspectives by showing that the ``average of shadows'' principle holds for general convex bodies in any dimension, with integrands that are directly interpretable as measures of central projections.

The remainder of the paper is organized as follows. In Section~\ref{s:prelim}, we present the mathematical background and notation needed for our analysis, including the Funk and Hilbert geometries, the Holmes--Thompson notions of volume and surface area, and their projective properties. In Section~\ref{s:vertex-decomp-ptope}, we derive the vertex decomposition formula for polytopes (Theorem~\ref{thm:cauchy-funk-polytope}). In Section~\ref{s:cauchy-gen}, we prove the Cauchy formula for arbitrary convex bodies (Theorem~\ref{thm:cauchy-funk-gen}) and derive the associated Funk--Crofton formula. Finally, in Section~\ref{s:other-geometries}, we discuss connections with other geometries, showing that our results yield exact surface area formulas for planar Hilbert geometries and hyperbolic space, as well as the exact Minkowski--Cauchy formula. In an appendix, we give an elementary proof of the duality of the Funk surface area, based on the vertex decomposition formula.

%=======================================================================
\section{Preliminaries} \label{s:prelim}
%=======================================================================

Let us begin by presenting notation and terminology that will be used throughout the paper. We use $\inner{\cdot}{\cdot}$ for the standard inner (dot) product on $\RE^d$, and $\|\cdot\| = \sqrt{\inner{\cdot}{\cdot}}$ for the Euclidean norm. Let $O$ denote the origin, let $B^d$ be the Euclidean unit ball centered at $O$, and let $S^{d-1} = \bd B^d$ be the unit sphere. Given a linear subspace $E \subset \RE^d$ and a set $K \subset \RE^d$, let $\orthproj{K}{E}$ denote the orthogonal projection of $K$ onto $E$.

A \emph{convex body} in $\RE^d$ is a compact convex set with nonempty interior. Given a convex set $K$, we denote its boundary and interior by $\bd K$ and $\interior(K)$, respectively. For $\alpha \geq 0$, $\alpha K$ denotes the dilation of $K$ by factor $\alpha$ about the origin, and for $x \in \RE^{d}$, $K + x$ denotes the translate of $K$ by $x$. Given a convex body $K$ and a direction $u \in S^{d-1}$, the \emph{support function} of $K$ is  $h_K(u) = \sup \{ \inner{u}{x} : x \in K \}$, and the corresponding supporting hyperplane is 
\[
    H_K(u) ~ = ~ \{ x \in \RE^d : \inner{u}{x} = h_K(u) \}. 
\]

For $1 \leq k \leq d$, let $\lambda_k$ denote the $k$-dimensional Hausdorff measure on $\RE^d$. In particular, $\lambda_d$ is the Lebesgue measure, and on any $C^1$ $k$-dimensional submanifold of $\RE^d$, $\lambda_k$ agrees with the usual $k$-dimensional surface area. Let $\sigma$ denote the standard rotation-invariant surface measure on $S^{d-1}$. Finally, let $\omega_d = \lambda_d(B^d)$. 

Given two $C^1$ hypersurfaces $\Phi,\Psi \subset \RE^d$ and a differentiable map $f:\Phi \to \Psi$, the Jacobian of $f$ at $x \in \Phi$ is the local area-stretch factor induced by the differential $Df(x)$ on the tangent space of $\Phi$ at $x$. Equivalently, it is the factor appearing in the change-of-variables formula for surface integrals (see, for example,~\cite{Car76}). 

Given two convex sets $A, B \subset \RE^d$, their \emph{Minkowski sum} is 
\[
    A + B ~ = ~ \{ a + b \ST a \in A,\ b \in B \}. 
\] 
Given a convex body $K \subset \RE^d$, its \emph{difference body} is 
\[
    \Delta(K) ~ = ~ K + (-K), 
\]
which is centrally symmetric (see Figure~\ref{f:preliminaries}(a)). 

%-----------------------------------------------------------------------
\begin{figure}[htbp]
  \centerline{\includegraphics[scale=0.40]{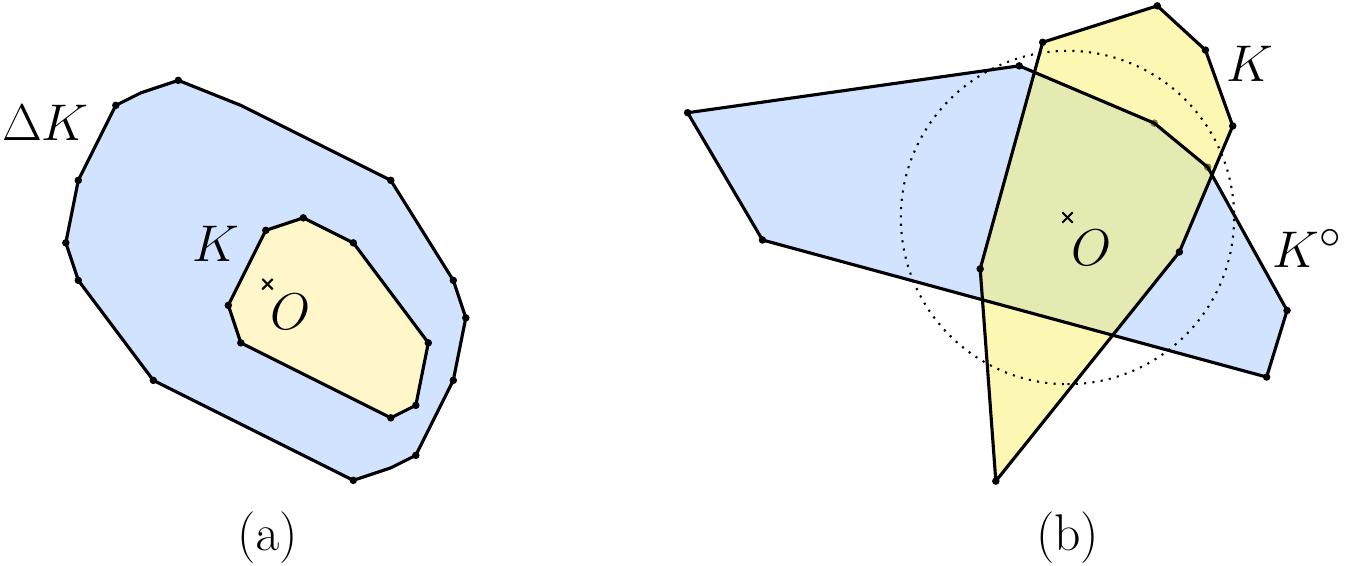}}
  \caption{(a) The difference body of a convex body and (b) the polar of a convex body.}
  \label{f:preliminaries}
\end{figure}
%-----------------------------------------------------------------------

For any nonempty set $K \subset \RE^d$, its \emph{polar}, denoted $\polar{K}$, is defined by 
\[
    \polar{K}
        ~ = ~ \{ y \in \RE^d \ST \inner{x}{y} \leq 1 \text{ for all } x \in K \}
\]
(see Figure~\ref{f:preliminaries}(b)). The polar has several standard properties (see, for example, Barvinok~\cite{Bar02}). It is always closed, convex, and contains the origin. If $O \in \interior(K)$, then $\polar{K}$ is bounded. Moreover, if $K$ is a convex body with $O \in \interior(K)$, then $\polar{K}$ is also a convex body, and $\polar{(\polar{K})} = K$. Finally, polarity reverses inclusion: if $K_1 \subseteq K_2$, then $\polar{K_2} \subseteq \polar{K_1}$. 

For a linear subspace $E \subseteq \RE^d$ and a nonempty set $G \subseteq E$, define the polar of $G$ in $E$ by
\[
    \polarIn[E]{G}
        ~ = ~ \{y\in E : \inner{x}{y} \leq 1, \text{~for all $x\in G$}\}.
\]
The next lemma is the standard duality between sections and projections, stated in the form needed later. A proof may be found in~\cite[Theorem 2.2.9 and Corollary 2.2.10]{Tho96}. 

%-----------------------------------------------------------------------
\begin{lemma}[Projection-Section Duality] \label{lem:slice}
Let $K \subseteq \RE^d$ be a convex body such that $O \in \interior(K)$, and let $E$ be a linear subspace of $\RE^d$. Then the polar in $E$ of the section $K \cap E$ equals the orthogonal projection of $\polar{K}$ onto $E$. That is,
\[
    \polarIn[E]{(K \cap E)} 
        ~ = ~ \orthproj{\polar{K}}{E}.
\]
\end{lemma}
%-----------------------------------------------------------------------

%-----------------------------------------------------------------------
\subsection{The Funk and Hilbert Geometries} \label{s:funk-hilbert}
%-----------------------------------------------------------------------

In this section, we introduce the Funk and Hilbert geometries and the related Finsler structures used to define volume and surface area. Let $K$ be a convex body in $\RE^d$. For any two distinct points $x, y \in \interior(K)$, let $y'$ denote the point where the ray directed from $x$ through $y$ intersects $\bd K$ (see Figure~\ref{f:funk-hilbert}(a)). The \emph{Funk distance} with respect to $K$, denoted $\distFunk[K](\cdot,\cdot)$, is defined to be
\[
    \distFunk[K](x,y)
       ~ = ~ \ln \frac{\|x - y'\|}{\|y - y'\|},
\]
where $\distFunk[K](x,y) = 0$ if $x = y$. The Funk distance is nonnegative, asymmetric, satisfies the triangle inequality, and is invariant under invertible affine transformations~\cite{PaT14}. Because of its asymmetry, this is often referred to as a weak metric or quasi-metric.

%-----------------------------------------------------------------------
\begin{figure}[htbp]
  \centerline{\includegraphics[scale=0.40,page=1]{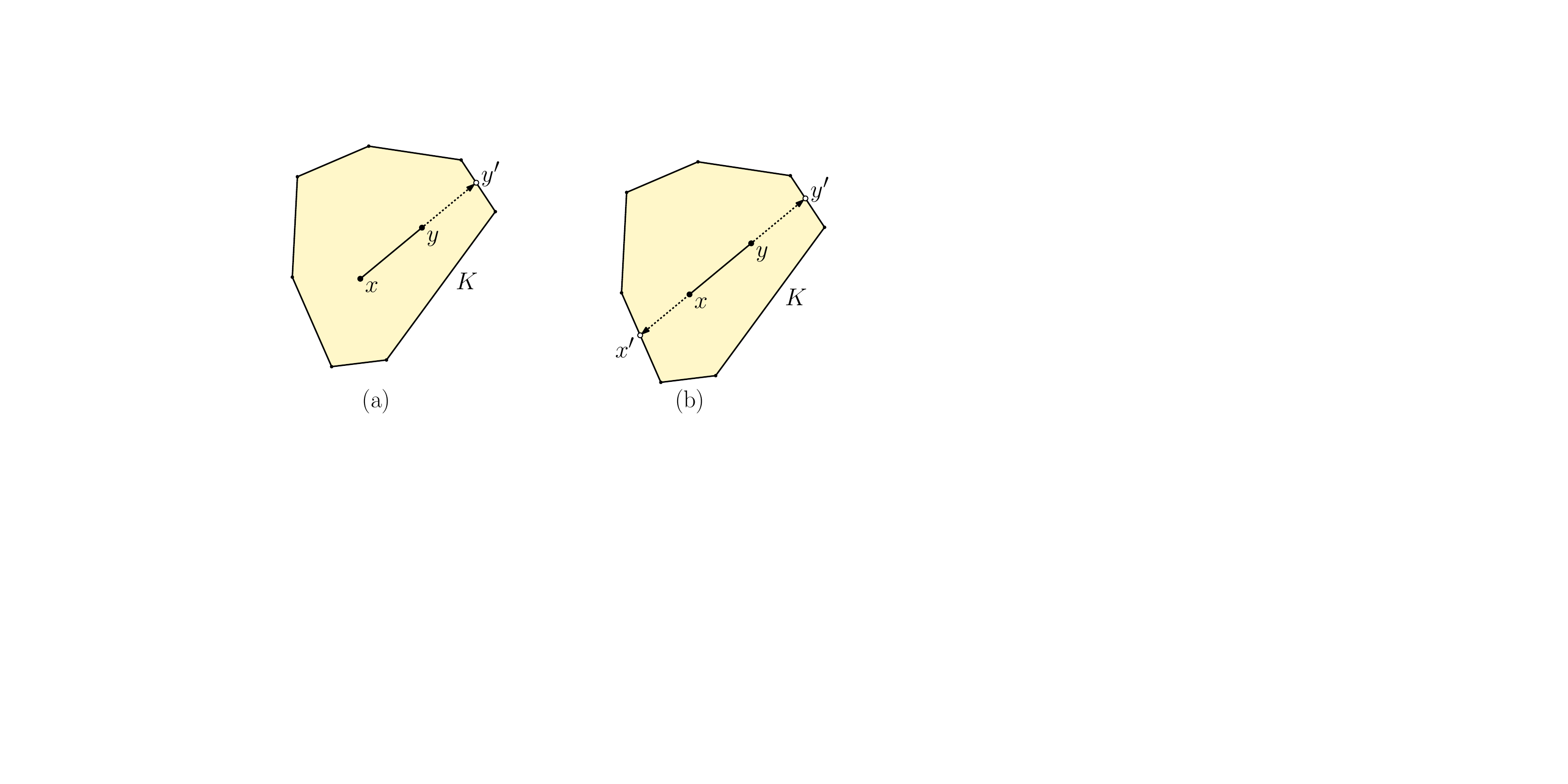}}
  \caption{(a) The Funk distance and (b) the Hilbert distance.}
  \label{f:funk-hilbert}
\end{figure}
%-----------------------------------------------------------------------

The \emph{Hilbert distance} with respect to $K$, denoted $\distHilb[K](\cdot,\cdot)$, is the symmetrization of the Funk distance. Letting $x'$ denote the point where the directed ray from $y$ through $x$ intersects $\bd K$, it is defined to be
\[
    \distHilb[K](x,y) 
        ~ = ~ \inv{2} \left( \distFunk[K](x,y) + \distFunk[K](y,x) \right) 
        ~ = ~ \inv{2} \ln \frac{\|y-x'\|}{\|x-x'\|} \frac{\|x-y'\|}{\|y-y'\|}
\]
(see Figure~\ref{f:funk-hilbert}(b)). It is well known that this is a metric (symmetric, positive except when $x=y$, and satisfying the triangle inequality). Observe that the quantity in the logarithm is the cross ratio $(x, y; y', x')$, and hence the Hilbert distance is invariant under invertible projective transformations. 

Before defining volume and surface areas, let us first introduce some related concepts. Given a convex body $D$ in $\RE^d$ with $O \in \interior(D)$ and $u \in \RE^d$, the \emph{Minkowski functional} (or \emph{gauge}) induced by $D$ is defined by
\[
    \|u\|_D ~ = ~ \inf \{\lambda > 0 \ST u \in \lambda D\}.
\]
This functional is positively homogeneous and subadditive. If $D$ is centrally symmetric, then $\|\cdot\|_D$ is a norm in the usual sense; in general, it is the Minkowski functional (or gauge) of $D$. A \emph{Finsler metric} on a manifold is a continuous function on its tangent bundle, whose restriction to each tangent space is such a gauge. It is well known that both the Funk and Hilbert geometries induce a Finsler structure~\cite{Tro14}.

To define the Finsler structure for Funk, consider any $x \in \interior(K)$. Identify the tangent space at $x$, denoted $T_x$, with $\RE^d$. For each nonzero $v \in T_x$, let $x^+$ denote the point where the ray emanating from $x$ in direction $v$ intersects $\bd K$ (see Figure~\ref{f:finsler}(a)), and let $t$ be a positive scalar such that $x + t v = x^+$. The resulting Finsler gauge at $x$ is defined as 
\[
    \finsFunk[K](x,v)
        ~ = ~ \frac{1}{t} ~~\text{and}~~ \finsFunk[K](x,0) ~ = ~ 0.
\]
This defines a gauge whose unit ball, denoted $\ballFunk[K](x)$, is $K - x$. Thus, the ball is just a translate of $K$ such that $x$ coincides with the origin.

%-----------------------------------------------------------------------
\begin{figure}[htbp]
  \centerline{\includegraphics[scale=0.40]{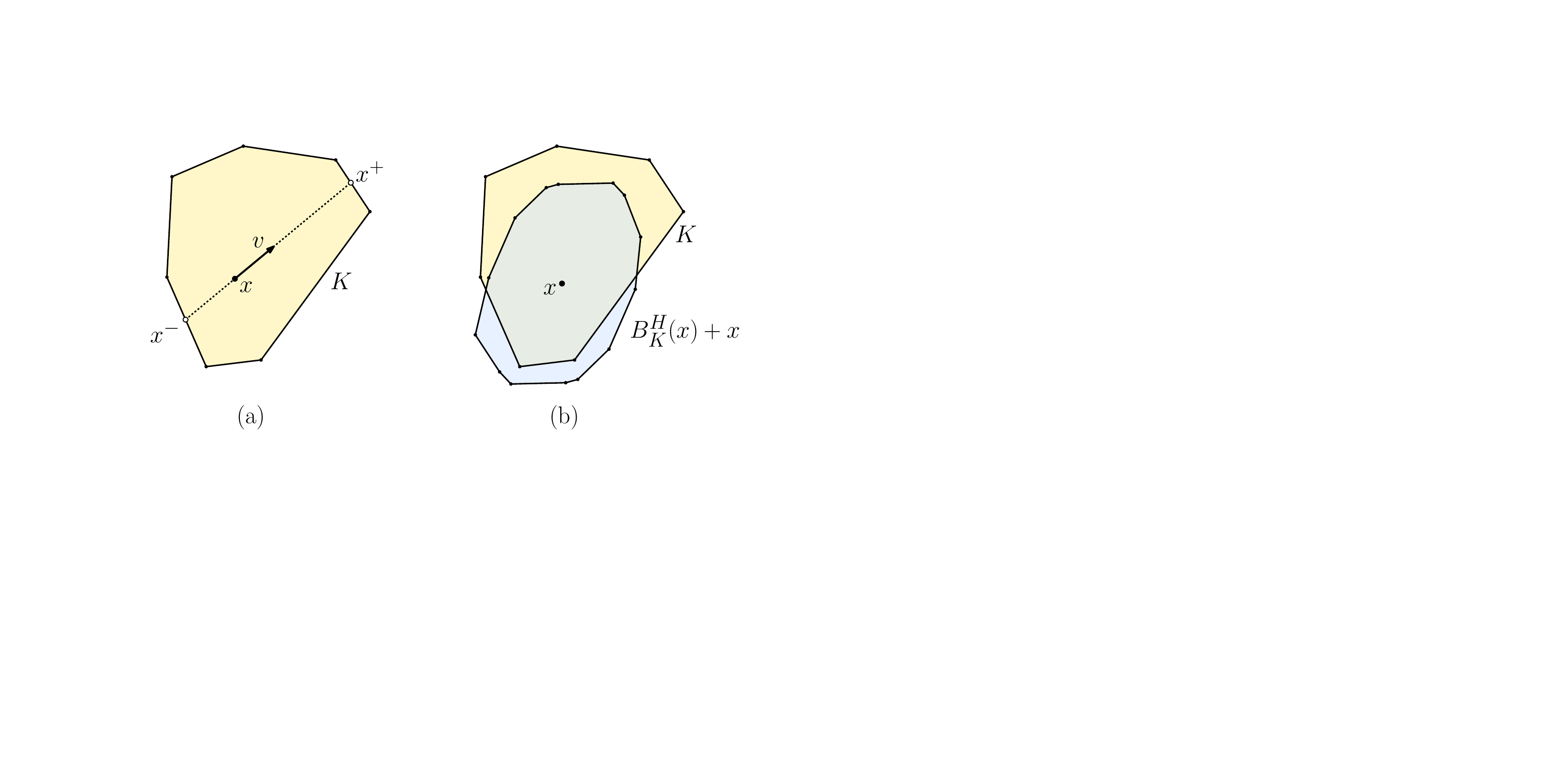}}
  \caption{(a) The Finsler structure and (b) the Finsler ball for Hilbert (recentered about $x$).}
  \label{f:finsler}
\end{figure}
%-----------------------------------------------------------------------

The Finsler structure for Hilbert is a symmetric variant of the Funk Finsler structure. For each $x \in \interior(K)$ and each nonzero vector $v \in T_x$, let $x^+$ be as before, and let $x^-$ be the point where the ray emanating from $x$ in direction $-v$ intersects $\bd K$ (see Figure~\ref{f:finsler}(a)). Let $t^+$ and $t^-$ be positive scalars such that $x + t^+ v = x^+$ and $x - t^- v = x^-$. Then
\[
    \finsHilb[K](x, v)
        ~ = ~ \frac{1}{2} \left( \frac{1}{t^+} + \frac{1}{t^-} \right) ~~\text{and}~~ \finsHilb[K](x,0) ~ = ~ 0.
\]
Since this gauge is symmetric, it is a norm. Its unit ball, denoted $\ballHilb[K](x)$, is centrally symmetric (see Figure~\ref{f:finsler}(b)). If $K$ is a polytope, then so is $\ballHilb[K](x)$.

The following lemma relates the polar of the Finsler ball in the Hilbert metric to the difference body of $\polar{(K-x)}$. The lemma is a straightforward consequence of the Finsler interpretation of the Hilbert metric (see, e.g.,~\cite{PaT09, Sch06b}). 

%-----------------------------------------------------------------------
\begin{lemma} \label{lem:hilbert-ball}
Given a convex body $K \subset \RE^d$ and $x \in \interior(K)$, $\polar{{\ballHilb[K](x)}} = \frac{1}{2} \Delta(\polar{(K - x)})$.
\end{lemma}
%-----------------------------------------------------------------------

%-----------------------------------------------------------------------
\subsection{Holmes--Thompson Volume and Surface Area} \label{s:ht-vol-area}
%-----------------------------------------------------------------------

There are several ways to define volume and surface area in Finsler spaces. The volume of a subset $G \subset \interior(K)$ arises by integrating the weighted contributions of volume elements over $G$, where the weight of each element depends on the local geometry, as expressed through the Finsler ball. Intuitively, as the size of the Finsler ball decreases, the relative importance of the associated Lebesgue volume element increases. Due to the reciprocal nature of polarity, as the volume of the Finsler ball decreases, the volume of its polar increases. This gives rise to the following definitions, due to Holmes and Thompson~\cite{HoT79}.

Let us first consider the Funk geometry. Given a convex body $K$ and $x \in \interior(K)$, the Holmes--Thompson volume element is obtained by scaling the Lebesgue volume element $d\lambda_d(x)$ by the ratio of the Lebesgue volume of the polar of the Funk Finsler ball, $\lambda_d\left(\polar{{\ballFunk[K](x)}}\right)$, and the Lebesgue volume of the unit Euclidean ball, $\omega_d$~\cite{FVW23, HoT79}. Define the \emph{Funk volume element} to be
\[
    d \volFunk[K](x) 
        ~ = ~ \frac{1}{\omega_d} \lambda_d\left(\polar{{\ballFunk[K](x)}}\right) \, d \lambda_d(x).
\] 
Recall that $\ballFunk[K](x) = K-x$. The \emph{Funk volume} of any convex body $G \subset \interior(K)$ is
\[
    \volFunk[K](G) 
        ~ = ~ \int_{x \in G} d\volFunk[K](x).
\] 

The Holmes--Thompson surface area in the Funk geometry is defined analogously. At each smooth boundary point $x \in \bd G$, let $T_x$ denote the tangent space to $\bd G$ at $x$, viewed as a linear subspace of $\RE^d$ and equipped with the induced Hausdorff measure. The \emph{Funk surface-area element} is given by
\[
    d \areaFunk[K](x) 
        ~ = ~ \frac{1}{\omega_{d-1}} \lambda_{d-1} \left(\polarIn[T_x]{\big(\ballFunk[K](x) \cap T_x\big)}\right) \, d \lambda_{d-1}(x).
\]
By Lemma~\ref{lem:slice}, we can express this alternatively in terms of the projected polar as
\[
    d \areaFunk[K](x) 
        ~ = ~ \frac{1}{\omega_{d-1}} \lambda_{d-1}\left( \orthproj{\polar{{\ballFunk[K](x)}}}{T_x} \right) \, d \lambda_{d-1}(x).
\]
Since $\ballFunk[K](x) = K-x$, this is exactly
\begin{equation} \label{eq:funk-area-element}
    d \areaFunk[K](x) 
        ~ = ~ \frac{1}{\omega_{d-1}} \lambda_{d-1}\left( \orthproj{\polar{(K-x)}}{T_x} \right) \, d \lambda_{d-1}(x).
\end{equation}
The \emph{Funk surface area} of $G \subset \interior(K)$ is
\[
    \areaFunk[K](G) 
        ~ = ~ \int_{x \in \bd G} d \areaFunk[K](x).
\]

The Holmes--Thompson volume and surface area in the Hilbert geometry are defined analogously, replacing the Funk Finsler ball $\ballFunk[K](x)$ by the Hilbert Finsler ball $\ballHilb[K](x)$. We denote the resulting quantities by $\volHilb[K](G)$ and $\areaHilb[K](G)$.

Faifman showed that the Hilbert- and Funk-based volumes and surface areas are related up to factors that depend on dimension.

%-----------------------------------------------------------------------
\begin{lemma}[Faifman~\cite{Fai24}] \label{lem:mod-def}
Let $K \subset \RE^d$ be a convex body, let $G \subset \interior(K)$ be a convex body, and let $\beta(d) = \binom{2d}{d}/2^d$. Then
\begin{align*}
    \volFunk[K](G) 
        & ~\leq~ \volHilb[K](G) ~\leq~ \beta(d) \cdot \volFunk[K](G) \qquad\text{and} \\
    \areaFunk[K](G) 
        & ~\leq~ \areaHilb[K](G) ~\leq~ \beta(d-1) \cdot \areaFunk[K](G).
\end{align*}
\end{lemma}
%-----------------------------------------------------------------------

%-----------------------------------------------------------------------
\subsection{Cones} \label{s:cones}
%-----------------------------------------------------------------------

The standard Cauchy formula relies on orthogonal projections. We shall see that in the context of the Funk geometry induced by a convex body, the appropriate generalization uses central projections towards the body's boundary. Such projections naturally involve cones. Throughout this paper, a \emph{cone} in $\RE^d$ denotes a full-dimensional convex set closed under positive scaling (i.e., if $x$ is in the set, then $\lambda x$ is in the set for all $\lambda \geq 0$). We assume our cones are \emph{pointed}, meaning that they contain no lines. This implies that each cone has a unique \emph{apex} at the origin.

The \emph{conical hull} of a convex set $G \subset \RE^d$, denoted $\cone(G)$, is the smallest cone containing $G$, that is,
\[
    \cone(G) 
        ~ := ~ \{ \gamma x \ST x \in G ~\text{and}~ \gamma \geq 0 \}
\]
(see Figure~\ref{f:cone-1}(a)). Let $K \subset \RE^d$ be a closed, full-dimensional convex set. For any boundary point $v \in \bd K$, the \emph{normal cone} at $v$ consists of all outer normal vectors to $K$ at $v$:
\[
    \polar{K_v}
        ~ := ~ \left\{y \in \RE^d \ST \inner{y}{x-v} \leq 0, ~\text{for all}~ x \in K \right\}.
\]
(see Figure~\ref{f:cone-1}(b)). In the cases of primary interest to us, such as when $v$ is a vertex of a polytope or the apex of a pointed cone, the tangent cone $K_v = \cone(K-v)$ is pointed. In these settings, the normal cone is full-dimensional and coincides with the polar cone $\polar{K_v}$.

%-----------------------------------------------------------------------
\begin{figure}[htbp]
  \centerline{\includegraphics[scale=0.40,page=1]{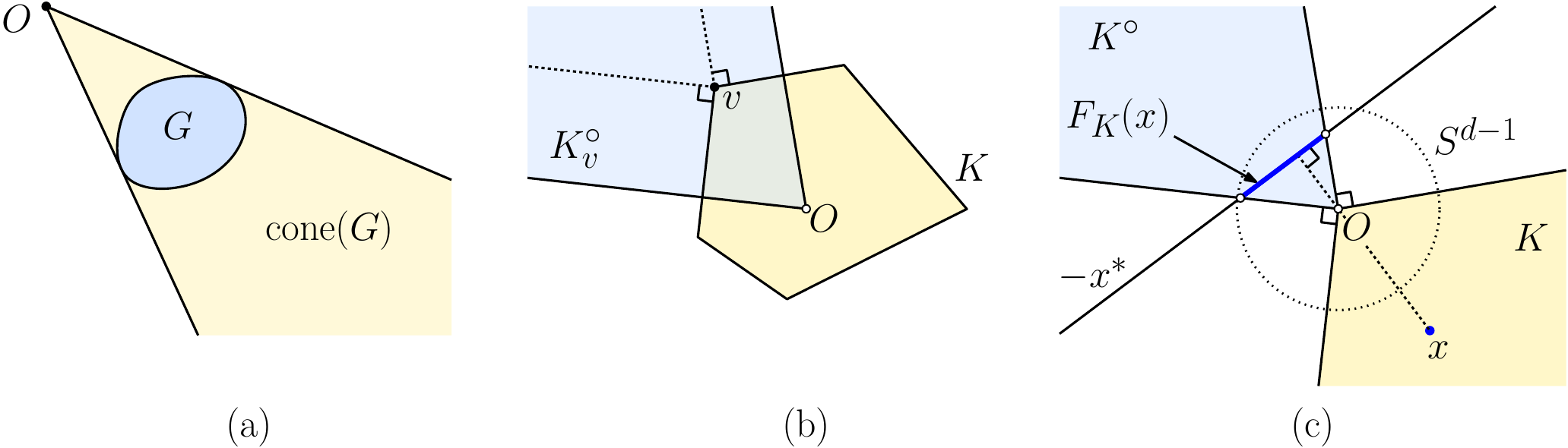}}
  \caption{(a) The cone $\cone(G)$, (b) the normal cone $\polar{K_v}$, and (c) $F_K(x)$.}
  \label{f:cone-1}
\end{figure}
%-----------------------------------------------------------------------

For any nonzero vector $x \in \RE^d$, we denote the associated \emph{dual hyperplane} by
\begin{equation} \label{eq:dual-hyperplane}
    x^* ~ = ~ \left\{y \in \RE^d \ST \inner{y}{x} = 1 \right\}.
\end{equation}
For a pointed cone $K \subset \RE^d$ and a vector $x \in \interior(K)$, we define the \emph{dual cross-section} $F_K(x)$ to be the intersection of the polar cone $\polar{K}$ with the dual hyperplane associated with $-x$ (see Figure~\ref{f:cone-1}(c)). That is,
\begin{equation} \label{eq:F-def}
    F_K(x) 
        ~ := ~ \polar{K} \cap -x^*
        ~ = ~ \polar{K} \cap (-x)^* 
        ~ = ~ \{ y \in \polar{K} : \inner{x}{y} = -1 \}.
\end{equation}
Geometrically, this is the slice of the polar cone by a hyperplane orthogonal to the direction $x$. Since $x$ lies in the interior of $K$, and every vector in $\polar{K}$ has a nonpositive inner product with every vector in $K$, this hyperplane intersects $\polar{K}$ in a bounded set.

Next, we establish a simple technical lemma relating the polar of a body to the polar of its subtended cones. This result allows us to characterize the facets of the polar polytope using local cone geometry.

%-----------------------------------------------------------------------
\begin{lemma} \label{lem:polar-slice}
Let $Q \subset \RE^d$ be a convex body containing the origin in its interior. Let $v \in \RE^d \setminus \interior(Q)$. Let $Q_v = \cone(Q-v)$ be the cone subtended by $Q$ at $v$ (see Figure~\ref{f:cone-4}(a)). Then the intersection of the polar cone $\polar{Q_v}$ with the dual hyperplane $v^*$ coincides with the slice of the polar body $\polar{Q}$ by $v^*$, that is,
\[
    \polar{Q_v} \cap v^* 
        ~ = ~ \polar{Q} \cap v^*
\]
(see Figure~\ref{f:cone-4}(b) and (c)).
\end{lemma}
%-----------------------------------------------------------------------

%-----------------------------------------------------------------------
\begin{figure}[htbp]
  \centerline{\includegraphics[scale=0.40,page=4]{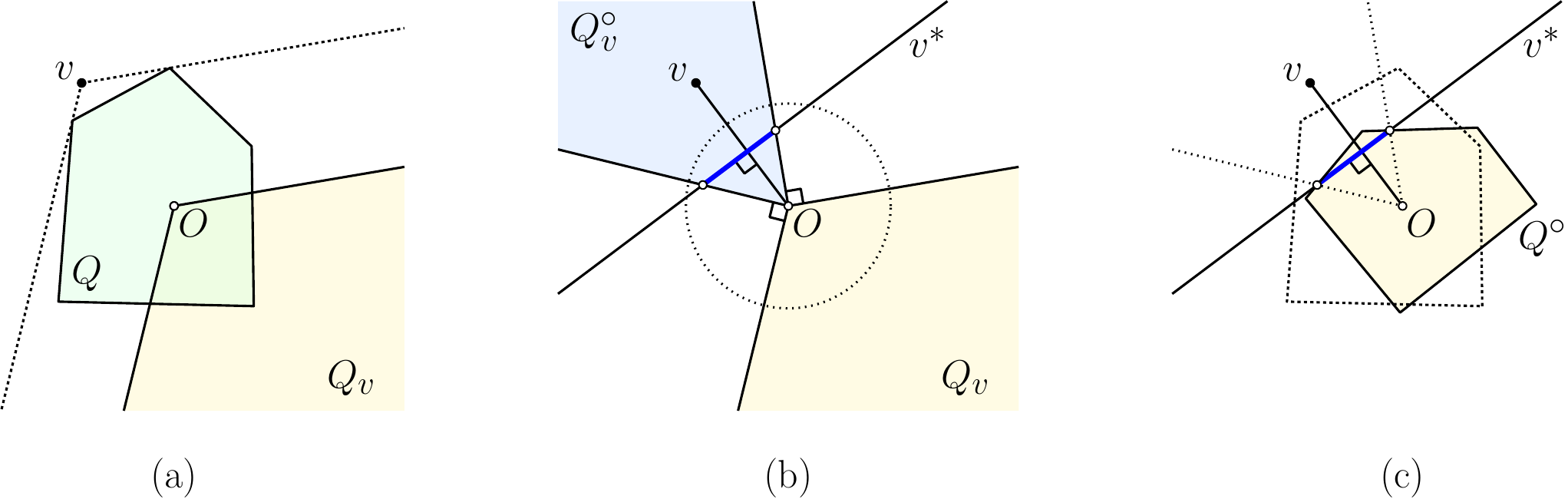}}
  \caption{(a) The cone $Q_v$, (b) $\polar{Q_v} \cap v^*$, and (c) $\polar{Q} \cap v^*$.}
  \label{f:cone-4}
\end{figure}
%-----------------------------------------------------------------------

%-----------------------------------------------------------------------
\begin{proof}
By definition, a vector $y \in \polar{Q_v}$ if and only if $\inner{z}{y} \le 0$ for all $z \in Q_v$. Since $Q_v$ is the conical hull of $Q-v$, this condition is equivalent to:
\[
    \inner{x-v}{y} ~ \leq ~ 0, \quad \forall x \in Q.
\]
Restricting our attention to the hyperplane $v^*$, where $\inner{v}{y}=1$, this inequality becomes:
\[
    \inner{x}{y} - \inner{v}{y} \leq 0 
        ~ \iff ~ \inner{x}{y} - 1 \leq 0 
        ~ \iff ~ \inner{x}{y} \leq 1, \quad \forall x \in Q.
\]
The condition $\inner{x}{y} \leq 1$ for all $x \in Q$ is exactly the definition of the polar body $\polar{Q}$. Thus, restricted to the hyperplane $v^*$, the condition for membership in the polar cone is identical to the condition for membership in the polar body.
\end{proof}
%-----------------------------------------------------------------------

When applied to the vertices of a polytope, this lemma provides a functional characterization of the boundary of the polar body. Recall that for a polytope $K$ with $O \in \interior(K)$, the facets of the polar body $\polar{K}$ are in one-to-one correspondence with the vertex set $V(K)$~\cite{Zie95}. Specifically, for each vertex $v \in V(K)$, the dual facet is the intersection of $\polar{K}$ with the supporting hyperplane $v^*$. Lemma~\ref{lem:polar-slice} implies that this facet is exactly $\polar{K_v} \cap v^*$, which matches the definition of the dual cross-section $F_{K_v}(-v)$.

%-----------------------------------------------------------------------
\begin{corollary} \label{cor:polar-facets}
Let $K \subset \RE^d$ be a convex polytope with $O \in \interior(K)$. For each vertex $v \in V(K)$, the facet of $\polar{K}$ dual to $v$, denoted $F_v$, is given by
\[
    F_v ~ = ~ \polar{K_v} \cap v^* ~ = ~ F_{K_v}(-v).
\]
Furthermore, $\bd \polar{K} = \bigcup_{v \in V(K)} F_v$.
\end{corollary}
%-----------------------------------------------------------------------

Extending the definition of Funk volume from bounded convex bodies to unbounded cones requires care. If $G$ and $K$ are pointed cones in $\RE^d$ with $(G \setminus \{O\}) \subset \interior(K)$, the standard Funk distance is degenerate, and the classical volume is infinite. However, the projective nature of Funk geometry allows us to define a meaningful volume for $G$ relative to $K$ by considering cross-sections. Faifman~\cite{Fai24} demonstrated that the Funk geometry is essentially projective. A key consequence is that the Holmes--Thompson volume of the section $G \cap H$, measured with respect to the ambient section $K \cap H$, is invariant to the choice of the hyperplane $H$, provided $K \cap H$ is bounded. We call such a hyperplane \emph{admissible}.

This invariance implies that the volume is intrinsic to the nested cone's structure. We define the \emph{cone Funk volume} of $G$ relative to $K$ by the integral
\begin{equation} \label{eq:funk-vol-cone-def}
     \volFunk[K](G)
        ~ := ~ \frac{1}{\omega_{d-1}} \int_{\Sigma(G)} \lambda_{d-1}(F_K(s)) \, d\sigma(s),
        \qquad\text{where $\Sigma(G) := S^{d-1} \cap G$}
\end{equation}
(see Figure~\ref{f:cone-2}(a)).  
In Lemma~\ref{lem:cone-vol} (see Appendix~\ref{s:just-cone-vol}), we show that this integral coincides with $\volFunk[K \cap H](G \cap H)$ for any admissible hyperplane $H$, thus justifying this terminology.

%-----------------------------------------------------------------------
\begin{figure}[htbp]
  \centerline{\includegraphics[scale=0.40,page=2]{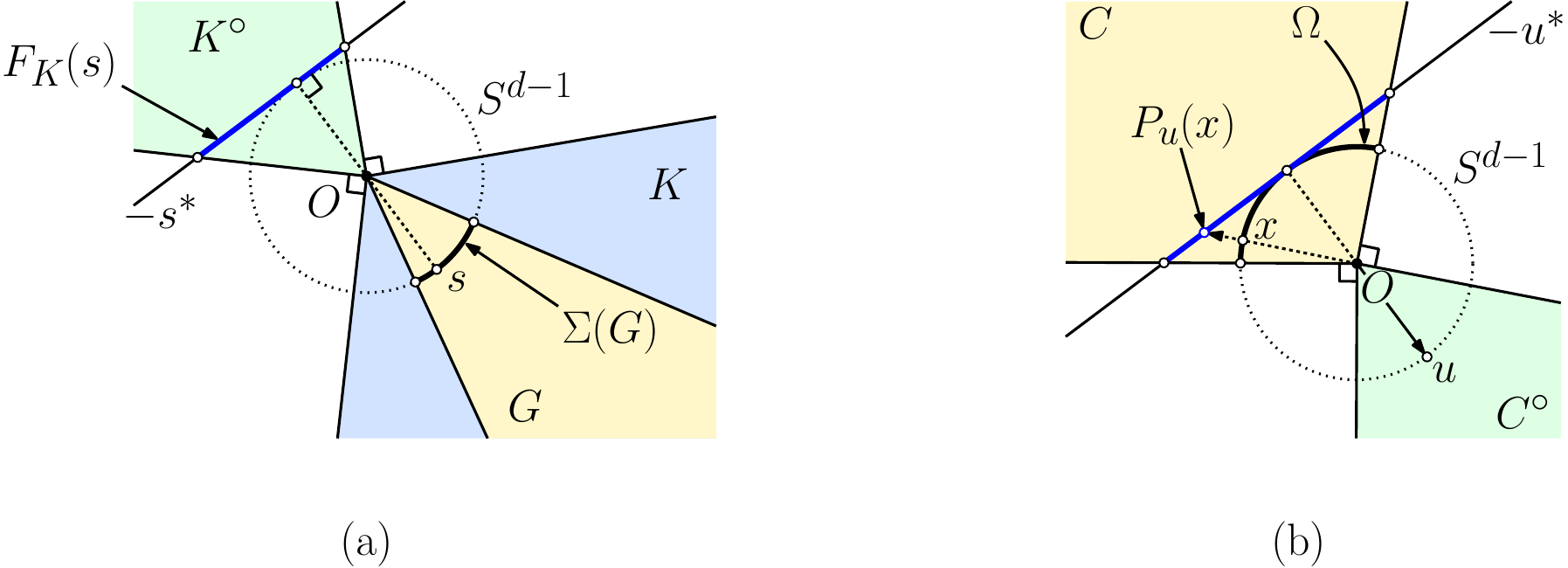}}
  \caption{(a) Defining $\volFunk[K](G)$ and (b) the gnomonic projection.}
  \label{f:cone-2}
\end{figure}
%-----------------------------------------------------------------------

When dealing with spherical cross-sections of cones, it is useful to relate their spherical measures to the Euclidean measures of corresponding hyperplane sections. Let $C \subset \RE^d$ be a pointed cone, let $\Omega = C \cap S^{d-1}$, and let $u \in S^{d-1}$ satisfy $\inner{x}{u} < 0$ for all $x \in \Omega$. The associated gnomonic projection maps $\Omega$ onto the hyperplane section $-u^* \cap C$ (see Figure~\ref{f:cone-2}(b)). A straightforward Jacobian computation yields the following lemma (see, e.g., \cite[Lemma~3.10]{Fai24}). 

%-----------------------------------------------------------------------
\begin{lemma}[Cone gnomonic projection] \label{lem:gnomonic}
Let $C \subset \RE^d$ be a pointed cone, let $\Omega = C \cap S^{d-1}$, and let $u \in S^{d-1}$ satisfy $\inner{x}{u} < 0$, for all $x \in \Omega$ (equivalently, $u \in \interior(\polar{C})$). Define
\[
    P_u(x) ~ := ~ -\frac{x}{\inner{x}{u}} \qquad (x \in \Omega).
\]
Then $P_u(\Omega) = -u^* \cap C$, and
\[
    \lambda_{d-1}(P_u(\Omega))
        ~ = ~ \int_{\Omega} |\inner{x}{u}|^{-d}\, d\sigma(x).
\]
\end{lemma}
%-----------------------------------------------------------------------

%=======================================================================
\section{Vertex Decomposition for Convex Polytopes} \label{s:vertex-decomp-ptope}
%=======================================================================

In this section, we establish a discrete surface area formula for the case where $K$ is a convex polytope. Our result provides a vertex-based decomposition of the Funk surface area similar in spirit to that of Alexander \etal~\cite{ABF05} in the planar Hilbert geometry, but derived by a distinct approach that applies in arbitrary dimensions.

Our proof proceeds by decomposing the Funk surface area of $G$ into local contributions associated with the vertices of the ambient polytope $K$. For each vertex $v \in V(K)$, let $K_v = \cone(K-v)$ and $G_v = \cone(G-v)$ denote the local cones subtended at $v$. We will show that the total Funk surface area of $G$ is the sum of the Funk volumes of the cones $G_v$ relative to $K_v$. We begin with a lemma that expresses the Funk volume of $\cone(G)$ relative to a pointed cone $K$ as a boundary integral involving the projected areas of the dual cross-sections $F_K(x)$. In the statement below, $T_x$ denotes the tangent space of $G$ at $x$.

%-----------------------------------------------------------------------
\begin{figure}[htbp]
  \centerline{\includegraphics[scale=0.40,page=3]{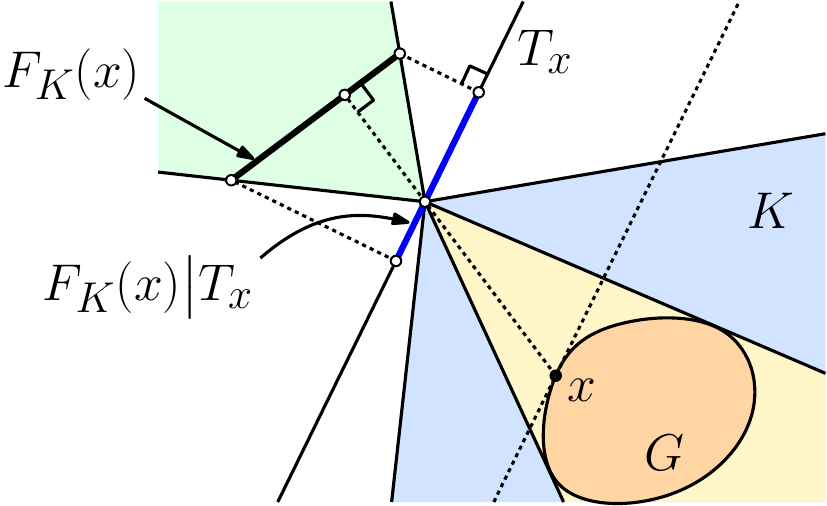}}
  \caption{Lemma~\ref{lem:cone-boundary-integral}.}
  \label{f:cone-3}
\end{figure}
%-----------------------------------------------------------------------

%-----------------------------------------------------------------------
\begin{lemma}\label{lem:cone-boundary-integral}
Let $K \subset \RE^d$ be a pointed cone and let $G \subset \interior(K)$ be a convex body. Then
\[
    2 \cdot \volFunk[K](\cone(G))
        ~ = ~ \frac{1}{\omega_{d-1}}\int_{\bd G} \lambda_{d-1} \big(\orthproj{F_K(x)}{T_x}\big) \, d\lambda_{d-1}(x)
\]
(see Figure~\ref{f:cone-3}).
\end{lemma}
%-----------------------------------------------------------------------

%-----------------------------------------------------------------------
\begin{proof}
Recall that 
\[
  \volFunk[K](\cone(G))
    ~ = ~ \frac{1}{\omega_{d-1}} \int_{\Sigma(G)} \lambda_{d-1}(F_K(s)) \, d\sigma(s), \qquad\text{where $\Sigma(G) = S^{d-1} \cap \cone(G)$.}
\]
Let $u_G(x)$ denote the unique outer unit normal at $x \in \bd G$ (well-defined almost everywhere~\cite[Theorem 2.2.5]{Sch14}). Consider the radial map $x \mapsto s := x / \|x\|$ from $\bd G$ to $\Sigma(G)$. The associated surface measure transformation is
\begin{equation}  \label{eq:measure-transform}
    d\sigma(s) 
        ~ = ~ \frac{|\inner{s}{u_G(x)}|}{\|x\|^{d-1}} \, d\lambda_{d-1}(x).
\end{equation}
(Geometrically, the term $|\inner{s}{u_G(x)}|$ represents the cosine of the angle between the radial vector and the surface normal. This factor corrects for the ``tilt'' of the surface patch $\partial G$ relative to the radial projection onto the sphere.)

To relate the integrands, we invoke the homogeneity of the construction. Recall that $F_K(x) = \polar{K} \cap -x^*$. Since $(\alpha x)^* = \alpha^{-1} x^*$, we have $F_K(x) = \|x\|^{-1} F_K(s)$. Therefore,
\begin{equation} \label{eq:scaling}
    \lambda_{d-1}(F_K(s)) 
        ~ = ~ \|x\|^{d-1} \, \lambda_{d-1}(F_K(x)).
\end{equation}
Further, the projection of $F_K(x)$ (which lies in a hyperplane orthogonal to $x$) onto $T_x$ (orthogonal to $u_G(x)$) scales its $(d-1)$-measure by $|\inner{s}{u_G(x)}|$, which implies that
\begin{equation} \label{eq:projection}
    \lambda_{d-1} (\orthproj{F_K(x)}{T_x}) 
        ~ = ~ \left| \inner{s}{u_G(x)} \right| \, \lambda_{d-1}(F_K(x)).
\end{equation}
Combining Eqs.~\eqref{eq:measure-transform}--\eqref{eq:projection} yields
\[
    \lambda_{d-1}(\orthproj{F_K(x)}{T_x}) \, d \lambda_{d-1}(x) 
        ~ = ~ \lambda_{d-1}(F_K(s)) \, d \sigma(s).
\]
Since $G$ is a convex body strictly contained in the cone $K$ (and thus does not contain the apex $O$), any ray from the origin that intersects $\interior(G)$ pierces $\bd G$ in exactly two points. Thus, the map $x \mapsto s$ is generically $2$-to-$1$ from $\bd G$ onto $\Sigma(G)$. Integrating over $\bd G$ yields twice the integral over $\Sigma(G)$, that is,
\[
    \int_{\bd G} \lambda_{d-1} \big(\orthproj{F_K(x)}{T_x}\big) \, d\lambda_{d-1}(x)
        ~ = ~ 2 \, \int_{\Sigma(G)} \lambda_{d-1}(F_K(s)) \, d \sigma(s).
\]
Dividing by $\omega_{d-1}$ completes the proof.
\end{proof}
%-----------------------------------------------------------------------

The next step in our derivation uses the dual facet characterization (Corollary~\ref{cor:polar-facets}) to decompose the polar body's projection into a sum over its facets. 

%-----------------------------------------------------------------------
\begin{lemma} \label{lem:facets-polar}
Let $K \subset \RE^d$ be a convex polytope. For any $x \in \interior(K)$ and any $(d-1)$-dimensional linear subspace $T$,
\[
    2 \lambda_{d-1} \big( \orthproj{\polar{(K-x)}}{T} \big)
        ~ = ~ \sum_{v\in V(K)} \lambda_{d-1} \big( \orthproj{F_{K_v}(x-v)}{T} \big)
\]
(see Figures~\ref{f:facets-polar}(a) and~(b)).
\end{lemma}
%----------------------------------------------------------------------- 

%----------------------------------------------------------------------- 
\begin{proof}
Let $Q=K-x$. Since $x \in \interior(K)$, we have $O \in \interior(Q)$, and the vertices of $Q$ are exactly $u=v-x$, for $v \in V(K)$. For each such $u$,
\[
    Q_u ~ = ~ \cone(Q-u) ~=~ \cone\big((K-x)-(v-x)\big) ~=~ \cone(K-v) ~=~ K_v.
\]
Hence, by Corollary~\ref{cor:polar-facets}, the facet of $\polar{Q}=\polar{(K-x)}$
dual to $u$ is
\[
    F_{Q_u}(-u) ~ = ~ F_{K_v}(x-v).
\]
Thus the sets $F_{K_v}(x-v)$, for $v \in V(K)$, are exactly the facets of $P := \polar{(K-x)}$.

For almost every $y \in \orthproj{P}{T}$, the line orthogonal to $T$ through $y$ meets $\bd P$ in exactly two points, each in the relative interior of a unique facet of $P$. Therefore, the projections of the facets of $P$ cover $\orthproj{P}{T}$ with multiplicity $2$ almost everywhere, and summing their projected $(d-1)$-measures yields the stated identity.
\end{proof}
%-----------------------------------------------------------------------

%-----------------------------------------------------------------------
\begin{figure}[htbp]
  \centerline{\includegraphics[scale=0.40]{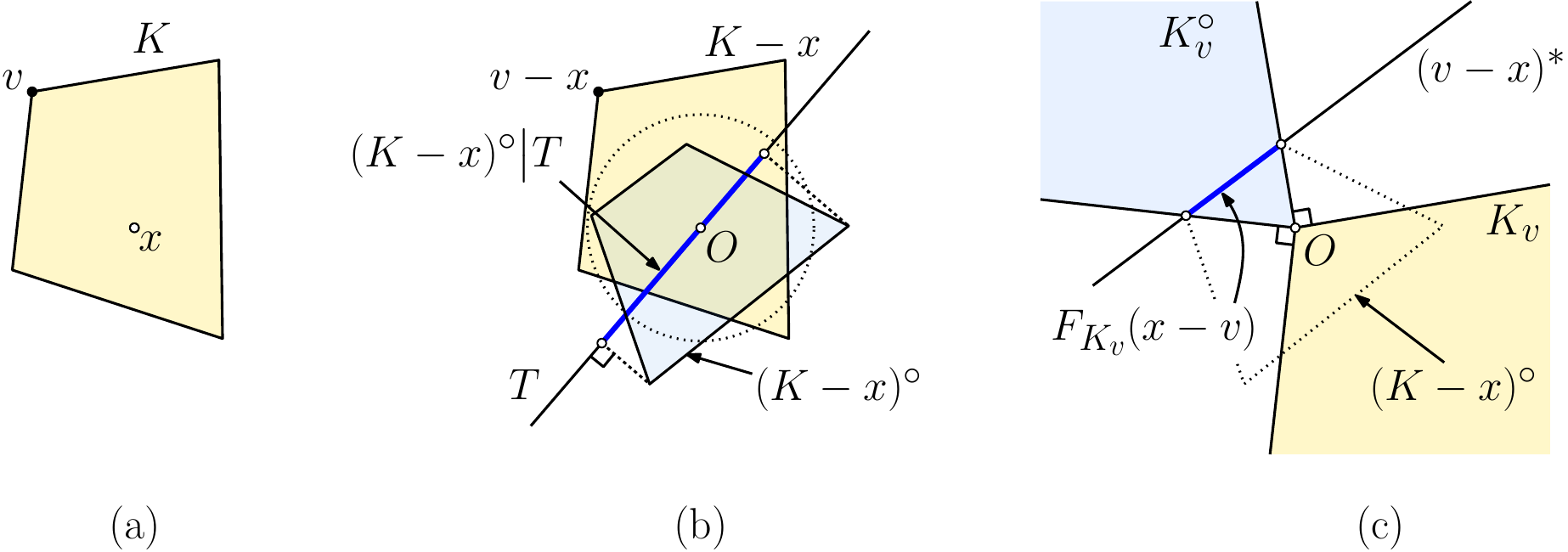}}
  \caption{Lemma~\ref{lem:facets-polar} and its proof.}
  \label{f:facets-polar}
\end{figure}
%-----------------------------------------------------------------------

We are now ready to prove the vertex decomposition theorem from Section~\ref{s:intro}.

%-----------------------------------------------------------------------
\begin{proof} (Of Theorem~\ref{thm:cauchy-funk-polytope})
For each $v\in V(K)$, applying Lemma~\ref{lem:cone-boundary-integral} to the pair $(K_v,G-v)$ yields
\[
    2 \cdot \volFunk[K_v](G_v)
        ~ = ~ \frac{1}{\omega_{d-1}} \int_{\bd(G-v)} \lambda_{d-1} \left(\orthproj{F_{K_v}(x')}{T_{x'}}\right) \, d\lambda_{d-1}(x').
\]
Perform the change of variables $x = x' + v$. Then $\bd(G-v)$ maps to $\bd G$, and the tangent spaces are identified as $T_{x'} = T_x$, yielding
\[
    2 \cdot \volFunk[K_v](G_v)
        ~ = ~ \frac{1}{\omega_{d-1}} \int_{\bd G} \lambda_{d-1} \left(\orthproj{F_{K_v}(x-v)}{T_x}\right) \, d\lambda_{d-1}(x).
\]
Summing over $v \in V(K)$ and using linearity of integration, we have
\[
    2 \sum_{v\in V(K)} \volFunk[K_v](G_v)
        ~ = ~ \frac{1}{\omega_{d-1}}\int_{\bd G} \left[ \sum_{v\in V(K)}\lambda_{d-1} \left(\orthproj{F_{K_v}(x-v)}{T_x}\right) \right] \, d\lambda_{d-1}(x).
\]
By Lemma~\ref{lem:facets-polar}, the term in brackets is equal to $2 \, \lambda_{d-1}(\orthproj{\polar{(K-x)}}{T_{x}})$. Substituting this and dividing by $2$ yields
\[
    \sum_{v \in V(K)} \volFunk[K_v](G_v)
        ~ = ~ \frac{1}{\omega_{d-1}} \int_{\bd G} \lambda_{d-1} \left(\orthproj{\polar{(K-x)}}{T_x}\right) \, d\lambda_{d-1}(x).
\]
By Eq.~\eqref{eq:funk-area-element}, the right-hand side is exactly $\areaFunk[K](G)$.
\end{proof}
%-----------------------------------------------------------------------

%=======================================================================
\section{Cauchy-Type Formulas for General Convex Bodies} \label{s:cauchy-gen}
%=======================================================================

We now derive the general Cauchy formula from the polytope case. The key additional ingredient is a representation of the Funk volume of a cone as an average of the areas of its central shadows. Combined with Theorem~\ref{thm:cauchy-funk-polytope}, this yields the formula for polytopal $K$, and the general case then follows by approximation. 

%-----------------------------------------------------------------------
\begin{figure}[htbp]
  \centerline{\includegraphics[scale=0.40,page=2]{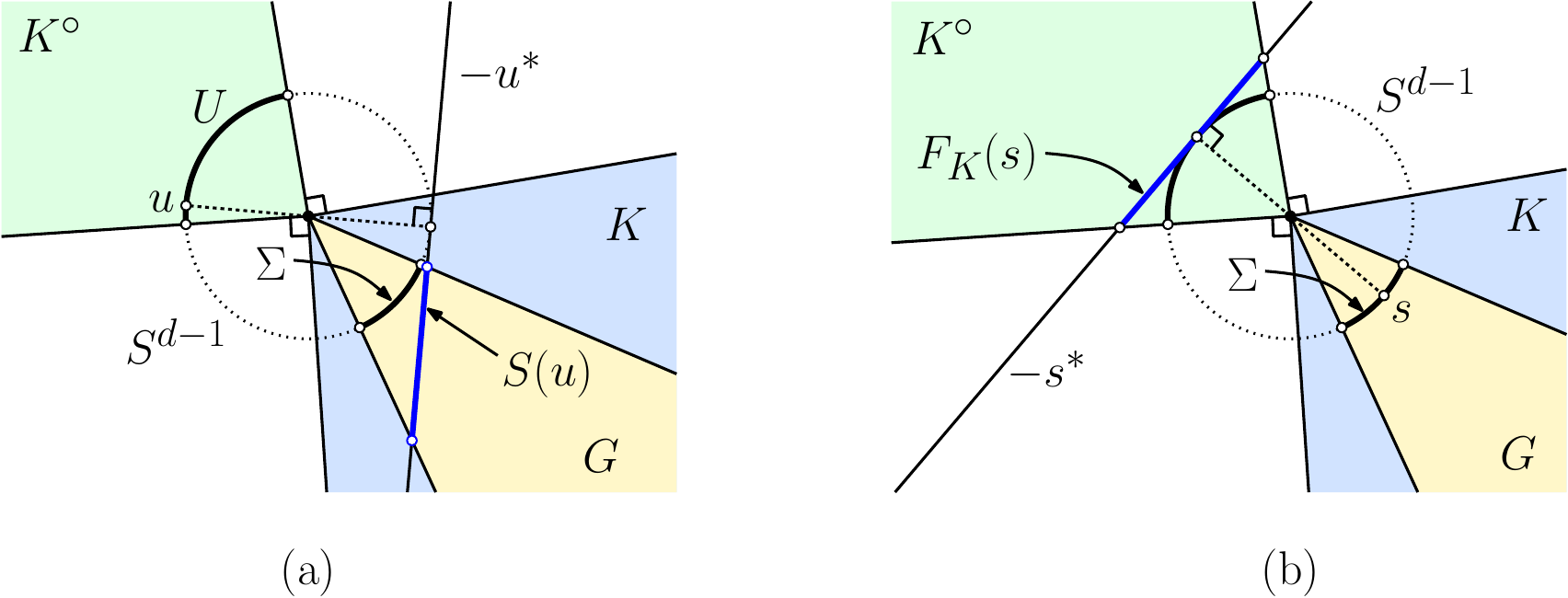}}
  \caption{Lemma~\ref{lem:funk-shadows} and its proof.}
  \label{f:funk-shadows-2}
\end{figure}
%-----------------------------------------------------------------------

%-----------------------------------------------------------------------
\begin{lemma}[Funk Volume of a Cone] \label{lem:funk-shadows}
Let $K$ and $G$ be pointed cones in $\RE^d$ with $(G \setminus \{O\}) \subset \interior(K)$. Let $U = \polar{K} \cap S^{d-1}$ (the spherical image of the polar cone), and for $u \in U$, let $S(u) = -u^* \cap G$ (see Figure~\ref{f:funk-shadows-2}(a)). Then
\begin{equation}
    \volFunk[K](G)
        ~ = ~ \frac{1}{\omega_{d-1}} \int_{U} \lambda_{d-1}(S(u)) \, d\sigma(u). \label{eq:funk-shadows}
\end{equation}
\end{lemma}
%-----------------------------------------------------------------------

%-----------------------------------------------------------------------
\begin{proof}
Let $\Sigma = S^{d-1} \cap G$ and define the double integral 
\[
    I ~ = ~ \frac{1}{\omega_{d-1}} \int_{U} \int_{\Sigma} \lvert \inner{s}{u} \rvert^{-d} \, d\sigma(s) \, d\sigma(u).
\]
Since the integrand is nonnegative, by Tonelli's theorem~\cite[Theorem~8.8(a)]{Rud87}, we may integrate in either order. We will show that one order leads to the right-hand side of Eq.~\eqref{eq:funk-shadows}, and the other leads to the left-hand side.

First fix $u \in U$ and consider the inner integral in $s$. Since $u \in \polar{K}$ and $(G \setminus \{O\}) \subset \interior(K)$, we have $\inner{s}{u} < 0$ for all $s \in \Sigma$. 
Applying Lemma~\ref{lem:gnomonic} to the cone $G$ with center $u$, we obtain 
\[
    \int_{\Sigma} \lvert \inner{s}{u} \rvert^{-d} \, d\sigma(s) 
        ~ = ~ \lambda_{d-1}(-u^* \cap G)
        ~ = ~ \lambda_{d-1}(S(u))
\]
(see Figure~\ref{f:funk-shadows-2}(a)).
Substituting into $I$ yields the right-hand side of Eq.~\eqref{eq:funk-shadows}.

Next fix $s \in \Sigma$ and consider the inner integral in $u$. Since $s \in \Sigma \subset \interior(K)$ and $U = \polar{K} \cap S^{d-1}$, we have $\inner{u}{s} < 0$ for all $u \in U$. 
Applying Lemma~\ref{lem:gnomonic} to the cone $\polar{K}$ with center $s$, we obtain 
\[
    \int_{U} \lvert \inner{u}{s} \rvert^{-d} \, d\sigma(u) 
        ~ = ~ \lambda_{d-1}(-s^* \cap \polar{K})
        ~ = ~ \lambda_{d-1}(F_K(s))
\]
(see Figure~\ref{f:funk-shadows-2}(b)).
Substituting into $I$, we obtain
\[
    I 
        ~ = ~ \frac{1}{\omega_{d-1}} \int_{\Sigma} \lambda_{d-1}(F_K(s)) \, d\sigma(s).
\]
By the definition of Funk volume (Eq.~\eqref{eq:funk-vol-cone-def}), this expression is exactly $\volFunk[K](G)$, which matches the left-hand side of Eq.~\eqref{eq:funk-shadows}.
\end{proof}
%-----------------------------------------------------------------------

With the cone formula established, we now derive the Cauchy formula for arbitrary convex bodies. Let $K \subset \RE^d$ be a convex body and let $G \subset \interior(K)$ be a convex body. For almost every direction $u \in S^{d-1}$, the support set $K \cap H_K(u)$ is a singleton; denote its unique point by $v_K(u)$~\cite[Theorem~2.2.11]{Sch14}. On the exceptional null set, choose $v_K(u)$ arbitrarily in $K \cap H_K(u)$. For each $u \in S^{d-1}$, define the central shadow by 
\[
    S_K(G,u) ~ := ~ -u^* \cap \cone(G-v_K(u))
\]
(recall Figure~\ref{f:funk-shadows-1}). We are now ready to prove Theorem~\ref{thm:cauchy-funk-gen}.  

%-----------------------------------------------------------------------
\begin{proof}[Proof of Theorem~\ref{thm:cauchy-funk-gen}]
Assume first that $K$ is a convex polytope. For each vertex $v \in V(K)$, let
\[
    U_v ~ := ~ \polar{K_v} \cap S^{d-1}
\]
denote the spherical image of the normal cone at $v$. The collection $\{U_v\}_{v \in V(K)}$ forms a partition of $S^{d-1}$ up to a set of $\sigma$-measure zero. For almost every $u \in U_v$, the support point of $K$ in direction $u$ is unique and equals $v$. Recalling that $G_v = \cone(G-v)$, it follows that for almost every $u \in U_v$,
\[
    S_{K_v}(G_v,u) ~ = ~ S_K(G,u).
\]
Therefore, by Theorem~\ref{thm:cauchy-funk-polytope} and Lemma~\ref{lem:funk-shadows}, we have
\begin{align*}
    \areaFunk[K](G)
        & ~ = ~ \sum_{v \in V(K)} \volFunk[K_v](G_v)
          ~ = ~ \sum_{v \in V(K)} \frac{1}{\omega_{d-1}}
        \int_{U_v} \lambda_{d-1}(S_{K_v}(G_v,u)) \, d\sigma(u) \\
        & ~ = ~ \frac{1}{\omega_{d-1}} \sum_{v \in V(K)}
        \int_{U_v} \lambda_{d-1}(S_K(G,u)) \, d\sigma(u)
          ~ = ~ \frac{1}{\omega_{d-1}}
        \int_{S^{d-1}} \lambda_{d-1}(S_K(G,u)) \, d\sigma(u).
\end{align*}
This establishes the result when $K$ is a polytope.

For a general convex body $K$, let $K_m$ be a sequence of convex polytopes converging to $K$ in the Hausdorff metric. Since $G \subset \interior(K)$ and $G$ is compact, after discarding finitely many terms we may assume that $G \subset \interior(K_m)$ for all $m$, and that there exists $\delta > 0$ such that the Euclidean distance from $G$ to $\partial K_m$ is at least $\delta$ for every $m$. Moreover, because $K$ is compact and $K_m \to K$ in the Hausdorff metric, the sets $K_m$ are all contained in a fixed Euclidean ball.

By the polytopal case,
\[
    \areaFunk[K_m](G)
        ~ = ~ \frac{1}{\omega_{d-1}} \int_{S^{d-1}} \lambda_{d-1}(S_{K_m}(G,u))\,d\sigma(u).
\]
Choose arbitrary support points $v_{K_m}(u) \in K_m \cap H_{K_m}(u)$ and  $v_K(u) \in K \cap H_K(u)$. For almost every $u \in S^{d-1}$, the support set $K \cap H_K(u)$ is a singleton. Since $K_m \to K$ in the Hausdorff metric, it follows that $v_{K_m}(u) \to v_K(u)$ for all such $u$. Hence $S_{K_m}(G,u) \to S_K(G,u)$ in the Hausdorff metric for almost every $u$, and therefore
\[
    \lambda_{d-1}(S_{K_m}(G,u))
        ~ \to ~ \lambda_{d-1}(S_K(G,u)), \quad\text{for almost every $u \in S^{d-1}$}.
\]
The uniform separation of $G$ from $\bd K_m$, together with the fact that the $K_m$ lie in a common Euclidean ball, implies that the shadows $S_{K_m}(G,u)$ are all contained in Euclidean balls of some fixed radius in $-u^*$, uniformly over $m$ and $u$. Hence $\lambda_{d-1}(S_{K_m}(G,u))$ is dominated by a constant independent of $m$ and $u$, and by the Dominated Convergence Theorem~\cite{Rud87}, we have
\[
    \frac{1}{\omega_{d-1}} \int_{S^{d-1}} \lambda_{d-1}(S_{K_m}(G,u))\,d\sigma(u)
        ~\to~
    \frac{1}{\omega_{d-1}} \int_{S^{d-1}} \lambda_{d-1}(S_K(G,u))\,d\sigma(u).
\]

It remains to pass to the limit in $\areaFunk[K_m](G)$. By Eq.~\eqref{eq:funk-area-element},
\[
    \areaFunk[K_m](G)
        ~ = ~ \frac{1}{\omega_{d-1}} \int_{\bd G} \lambda_{d-1}\!\left( \orthproj{\polar{(K_m-x)}}{T_x} \right) \,d\lambda_{d-1}(x).
\]
For each $x \in \bd G$, the uniform separation of $G$ from $\bd K_m$ implies that $K_m-x$ contains a fixed Euclidean ball about the origin, uniformly in $m$ and $x$. Hence, the projected polars in Eq.~\eqref{eq:funk-area-element} are uniformly contained in Euclidean balls of some fixed radius in $T_x$. Since $(K_m-x) \to (K-x)$ for each $x \in \bd G$, and all these sets contain a fixed Euclidean ball about the origin, polarity is continuous with respect to the Hausdorff metric on this family. Therefore the sets $\orthproj{\polar{(K_m-x)}}{T_x}$ converge in the Hausdorff metric to $\orthproj{\polar{(K-x)}}{T_x}$, and hence their $(d-1)$-dimensional volumes converge pointwise in $x$. Another application of the Dominated Convergence Theorem~\cite{Rud87} gives 
\[
    \areaFunk[K_m](G) ~ \to ~ \areaFunk[K](G).
\]
Combining these limits with the polytopal identity for $K_m$ completes the proof.
\end{proof}
%----------------------------------------------------------------------- 

From a computational perspective, when $K$ is a polytope, the preceding proof yields a natural Monte-Carlo estimator. The identity
\[
    \areaFunk[K](G)
    ~ = ~ \sum_{v \in V(K)} \frac{1}{\omega_{d-1}} \int_{U_v} \lambda_{d-1}(S_{K_v}(G_v, u)) \, d\sigma(u)
\]
decomposes the surface area into local contributions indexed by the vertices of $K$. The sampling distribution depends only on the spherical normal-cone decomposition $\{U_v\}_{v \in V(K)}$, while the sampled quantity is the $(d-1)$-dimensional volume of the corresponding central shadow. Thus, one may sample a vertex $v$ with probability proportional to $\sigma(U_v)$ and then sample a direction $u$ uniformly from $U_v$, for example, by triangulating the normal cone into simplicial cones. This yields an unbiased estimator whose evaluation uses only geometric information local to the chosen vertex. In particular, it avoids direct use of the Holmes--Thompson definition through Eq.~\eqref{eq:funk-area-element}, which would require integrating
\[
    \lambda_{d-1}\!\left(\orthproj{\polar{(K-x)}}{T_x}\right)
\]
over all $x \in \bd G$. 

\bigskip

Using Lemma~\ref{lem:gnomonic}, we can rewrite the shadow integral in Theorem~\ref{thm:cauchy-funk-gen} in terms of the spherical cross-sections of the subtended cones. For any $v \in \bd K$, let
\[
    \Sigma(G,v) ~ := ~ S^{d-1} \cap \cone(G - v)
\]
denote the spherical cross-section of the cone subtended by $G$ at $v$. The following corollary gives a double-integral representation of the Funk surface area. 

%-----------------------------------------------------------------------
\begin{corollary}[Double-Integral Formula for Funk Area] \label{cor:funk-double-int} 
Let $G$ and $K$ be convex bodies in $\RE^d$ with $G \subset \interior(K)$. Then
\[
    \areaFunk[K](G)
        ~ = ~ \frac{1}{\omega_{d-1}} \int_{S^{d-1}} \int_{\Sigma(G,v_K(u))} \lvert \inner{s}{u} \rvert^{-d} \, d\sigma(s) \, d\sigma(u).
\]
\end{corollary}
%-----------------------------------------------------------------------

%-----------------------------------------------------------------------
\begin{proof}
By Theorem~\ref{thm:cauchy-funk-gen},
\[
    \areaFunk[K](G)
        ~ = ~ \frac{1}{\omega_{d-1}} \int_{S^{d-1}} \lambda_{d-1}(S_K(G,u)) \, d\sigma(u).
\]
For each $u \in S^{d-1}$,
\[
    S_K(G,u) ~ = ~ -u^* \cap \cone(G - v_K(u)).
\]
Since $v_K(u)$ is a support point of $K$ in direction $u$ and $G \subset \interior(K)$,  we have $\inner{s}{u} < 0$ for all $s \in \Sigma(G,v_K(u))$. Applying Lemma~\ref{lem:gnomonic} with $C=\cone(G - v_K(u))$ and center $u$, we obtain 
\[
    \lambda_{d-1}(S_K(G,u))
        ~ = ~ \int_{\Sigma(G,v_K(u))} \lvert \inner{s}{u} \rvert^{-d}\, d\sigma(s).
\]
Substituting this into the preceding formula completes the proof. 
\end{proof}
%-----------------------------------------------------------------------

We now recast Corollary~\ref{cor:funk-double-int} as a Crofton formula on oriented lines. For each $u \in S^{d-1}$, choose $v_K(u) \in K \cap H_K(u)$ as above, and let
\[
    \mathcal P ~ := ~ \{(u,s)\in S^{d-1}\times S^{d-1} \ST \inner{s}{u}<0\}.
\]
For $(u,s)\in \mathcal P$, let $\ell_K^+(u,s)$ denote the oriented line through $v_K(u)$ in direction $s$.
For fixed $u \in S^{d-1}$, a direction $s \in S^{d-1}$ with $\inner{s}{u}<0$ belongs to $\Sigma(G,v_K(u))$ if and only if the ray from $v_K(u)$ in direction $s$ meets $G$. Thus Corollary~\ref{cor:funk-double-int}
can be rewritten as follows.

%-----------------------------------------------------------------------
\begin{lemma}[Funk--Crofton formula for oriented lines] \label{lem:crofton-oriented}
Let $G$ and $K$ be convex bodies in $\RE^d$ with $G \subset \interior(K)$. Then
\[
    \areaFunk[K](G)
        ~ = ~ \frac{1}{\omega_{d-1}} \int_{\mathcal P} \mathbf 1_{\{\ell_K^+(u,s)\cap G \neq \emptyset\}} \, |\inner{s}{u}|^{-d} \, d\sigma(s)\,d\sigma(u).
\]
\end{lemma}
%-----------------------------------------------------------------------

Replacing each oriented line by its underlying unoriented line, we obtain the following $K$-dependent parameter-space Crofton formula.

%-----------------------------------------------------------------------
\begin{theorem}[Funk--Crofton formula for unoriented lines] \label{thm:crofton-unoriented}
Let $G$ and $K$ be convex bodies in $\RE^d$ with $G \subset \interior(K)$. Then
\[
    \areaFunk[K](G)
        ~ = ~ \frac{1}{\omega_{d-1}} \int_{\mathcal P} \mathbf 1_{\{L_K(u,s)\cap G \neq \emptyset\}} \, |\inner{s}{u}|^{-d} \, d\sigma(s)\,d\sigma(u),
\]
where $L_K(u,s)$ denotes the unoriented line underlying $\ell_K^+(u,s)$.
\end{theorem}
%-----------------------------------------------------------------------

%-----------------------------------------------------------------------
\begin{proof}
Since intersection with $G$ is independent of orientation, the indicator in Lemma~\ref{lem:crofton-oriented} is unchanged when $\ell_K^+(u,s)$ is replaced by its underlying unoriented line $L_K(u,s)$.
\end{proof}
%-----------------------------------------------------------------------

%=======================================================================
\section{Connections to Other Geometries}  \label{s:other-geometries}
%=======================================================================

Our results on the Funk geometry can be applied to yield Cauchy-type formulas for other geometries as well. In this section, we explore some examples.

%=======================================================================
\subsection{Hilbert and Hyperbolic Geometries} \label{s:hilbert}
%=======================================================================

Let us first consider implications for the Hilbert geometry. The Holmes--Thompson surface area in the Hilbert geometry is generally difficult to compute because the Hilbert density involves the difference body of the polar set. In contrast, our Funk--Cauchy formula (Theorem~\ref{thm:cauchy-funk-gen}) relies solely on Euclidean shadows.

The following corollary shows that our formula gives a dimension-dependent approximation to the Hilbert area in general, and an \emph{exact} formula in two important cases: planar geometries and the Beltrami-Klein model of hyperbolic geometry in any dimension.

%-----------------------------------------------------------------------
\begin{corollary} \label{cor:hilbert-exact}
Let $G \subset \interior(K)$ be convex bodies in $\RE^d$.
\begin{enumerate}
\item[$(i)$] \emph{(General Approximation)} The Funk surface area approximates the Hilbert surface area within a constant factor depending only on dimension:
\[
    \areaFunk[K](G) 
        ~ \leq ~ \areaHilb[K](G) 
        ~ \leq ~ \beta(d-1) \cdot \areaFunk[K](G),
\]
where $\beta(d) := \binom{2d}{d}/2^d$.

\item[$(ii)$] \emph{(Planar Exactness)} If $d=2$, the measures coincide, that is,
\[
    \areaFunk[K](G) 
        ~ = ~ \areaHilb[K](G).
\]
\item[$(iii)$] \emph{(Ellipsoidal Exactness)} If $K$ is an ellipsoid, the measures coincide in any dimension:
\[
    \areaFunk[K](G) 
        ~ = ~ \areaHilb[K](G).
\]
This implies that exactness holds in the Beltrami-Klein model of hyperbolic geometry.
\end{enumerate}
\end{corollary}
%-----------------------------------------------------------------------

\begin{proof}
Part (i) follows directly from the comparison bounds in Lemma~\ref{lem:mod-def}. In the special case of $\RE^2$, we have $\beta(d-1) = \beta(1) = 1$. Thus, the upper and lower bounds match, which implies~(ii).

For Part (iii), let $K$ be an ellipsoid. For any $x \in \interior(K)$, the set $K-x$ is again an ellipsoid, with $O \in \interior(K-x)$. Hence, its polar $\polar{(K-x)}$ is also an ellipsoid, although in general it is not centered at the origin~\cite[Section 2.1]{KuV06}. By Lemma~\ref{lem:hilbert-ball},
\[
    \polar{{\ballHilb[K](x)}} 
        ~ = ~ \tfrac{1}{2} \Delta(\polar{(K-x)}).
\]

Now let $E$ be any ellipsoid and let $B^d$ denote the Euclidean unit ball in $\RE^d$. We first observe that $\tfrac{1}{2} \Delta(E)$ is a translate of $E$. We can write $E = c + A B^d$, where $c \in \RE^d$ and $A$ is invertible, from which we have
\[
    \Delta(E)
        ~ = ~ E + (-E)
        ~ = ~ (c + A B^d) + (-c + A B^d)
        ~ = ~ 2 A B^d
        ~ = ~ 2 E - 2 c.
\]
Thus $\Delta(E)$ is a translate of $2 E$, and therefore $\tfrac12\Delta(E)$ is a translate of $E$. Applying this to $E = \polar{(K-x)}$, we conclude that $\polar{{\ballHilb[K](x)}}$ is a translate of $\polar{{\ballFunk[K](x)}} = \polar{(K-x)}$. Since the Holmes--Thompson surface area density is determined by the volume of the orthogonal projection of the polar unit ball, and volume is translation-invariant, the corresponding densities agree:
\[
    d \areaHilb[K](x) ~ = ~ d \areaFunk[K](x),
\]
and hence, the associated surface areas coincide.
\end{proof}

This exactness explains and generalizes previous observations in the literature. Alexander \etal~\cite{ABF05} noted that for the Beltrami--Klein model of hyperbolic geometry in the plane, the Hilbert perimeter is the average length of central shadows. The underlying reason is that the shadow principle belongs naturally to the Funk geometry. By Theorem~\ref{thm:cauchy-funk-gen}, the Funk surface area is always given exactly by an average of central shadows in every dimension. Thus, the exact Hilbert shadow formulas in these cases are not isolated phenomena, but instances of the general Funk shadow formula in situations where the two geometries agree on surface area.

%=======================================================================
\subsection{Minkowski Geometry} \label{s:minkowski}
%=======================================================================

Our results also have implications for the Holmes--Thompson surface area in Minkowski geometry. Consider the Minkowski geometry induced by a convex body $K$ in $\RE^d$ containing the origin in its interior, and let $\|\cdot\|_K$ denote the associated Minkowski functional (or gauge). We denote the corresponding Holmes--Thompson surface area of $G$ by $\areaMink[K](G)$. Recall that the Holmes--Thompson surface area element in the Minkowski geometry induced by $K$ is defined analogously to the Funk and Hilbert cases, using the polar of the local unit ball. In the Minkowski setting, the unit ball is translation-invariant, that is, for every $x$, the local ball is the translate $x+K$. Its polar with respect to the local center is simply $\polar{K}$. Thus, the Holmes--Thompson surface area element is translation-invariant. For any $x \in \bd G$, let $T_x$ denote the tangent hyperplane at $x$. (We may ignore the set of measure zero where $T_x$ is not unique.) Recalling that $\orthproj{\polar{K}}{T_x}$ denotes the orthogonal projection of $\polar{K}$ onto $T_x$, the Minkowski Holmes--Thompson surface area element is
\[
    d\areaMink[K](x) 
        ~ = ~ \frac{1}{\omega_{d-1}} \lambda_{d-1}\left( \orthproj{\polar{K}}{T_x} \right) \, d\lambda_{d-1}(x)
\]
(see, e.g.,~\cite{Tho96}). As in Section~\ref{s:funk-hilbert}, the surface area of $G$ is defined by integrating this area element over $\bd G$.

Recall that in the Funk geometry, the fundamental geometric object in our Cauchy formula is the \emph{central shadow} $S_K(G, u)$, formed by projecting $G$ from a boundary point $v_K(u)$ onto the hyperplane $u^\perp$. Recall that $v_K(u) \in \partial K$ is a support point in direction $u$, so that $\inner{v_K(u)}{u}=h_K(u)$. The Minkowski case can be modeled by considering the limit as we scale $K$ up so that its boundary recedes to infinity. Central projections tend to parallel projections in the limit. To stress this structural parallel, we define the \emph{Minkowski shadow} of $G$ in direction $u$, denoted $S^M_K(G, u)$, to be the parallel projection of $G$ onto $u^\perp$ along the direction $v_K(u)$, followed by a scaling by a factor of $1/h_K(u)$ (see Figure~\ref{f:minkowski-shadows-1}). One way to interpret this is as the set of vectors $x \in u^\perp$ such that, after scaling by $h_K(u)$ and adding an appropriate scaled copy of $v_K(u)$, the resulting point lies in $G$. This leads to the definition:
\[
    S^M_K(G, u) 
        ~ := ~ \{ x \in u^\perp : \{h_K(u) \, x + t \, v_K(u) : t \in \RE \} \cap G \neq \emptyset \}
\]
Observe that since $O \in \interior(K)$, $h_K(u) > 0$ for all $u \in S^{d-1}$. Another way to express this is to recall from linear algebra that the projection of a point $g$ onto $u^\perp$ along direction $v$ is given by $g - \frac{\inner{g}{u}}{\inner{v}{u}} v$. Setting $v = v_K(u)$ and observing that $h_K(u) = \inner{v_K(u)}{u}$, we have $S^{\mathrm M}_K(G,u) = \Pi_u(G)$, where
\[
    \Pi_u(G) ~ := ~ \{\Pi_u(g):g\in G\} \quad\text{and}\quad
    \Pi_u(g) ~ := ~ \frac{1}{h_K(u)}\left(g-\frac{\inner{g}{u}}{h_K(u)}\,v_K(u)\right).
\]

%-----------------------------------------------------------------------
\begin{figure}[htbp]
  \centerline{\includegraphics[scale=0.40,page=1]{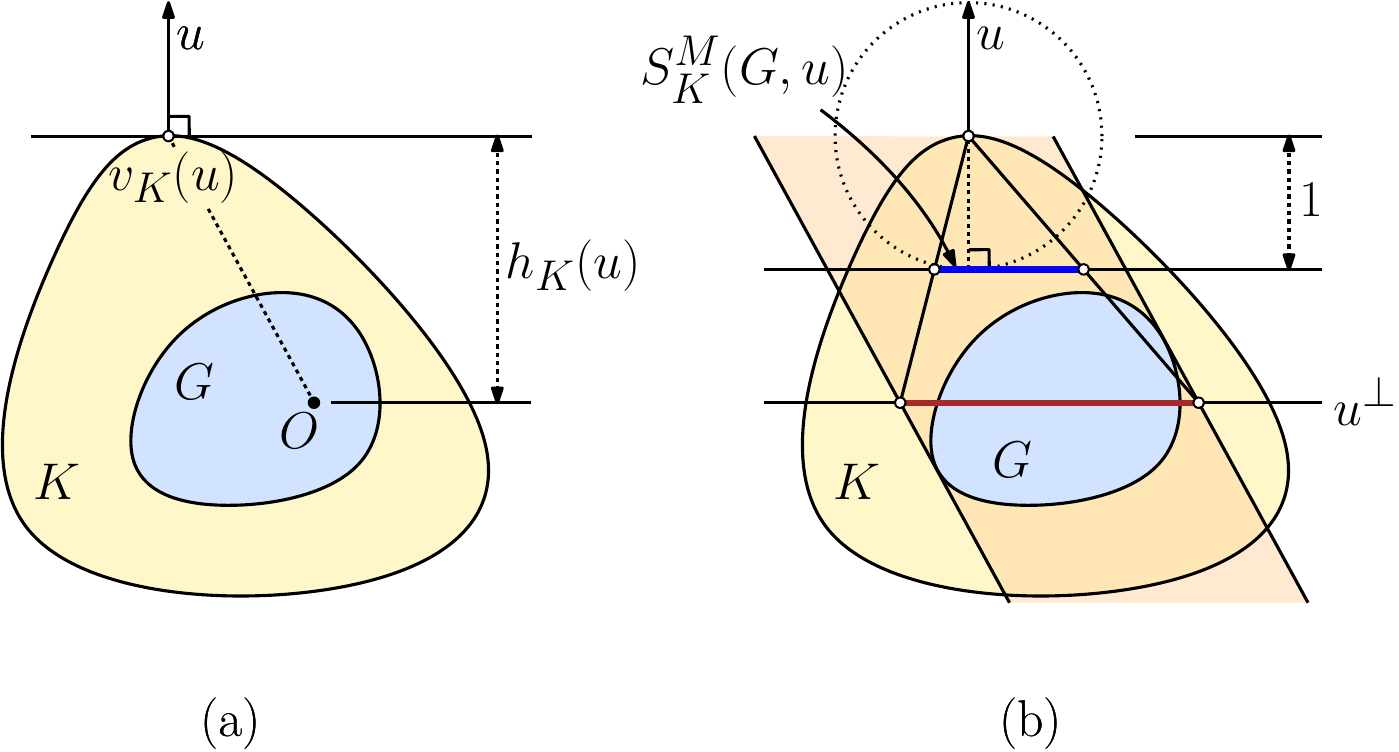}}
  \caption{(a) A body $G$ in the Minkowski geometry of $K$, and (b) the Minkowski shadow $S^M_K(G, u)$. (To better illustrate the analogy with central shadows, we have scaled and translated $G$ so it lies within $K$, but this need not be the case in general. We have also translated $S^M_K(G, u)$.)}
  \label{f:minkowski-shadows-1}
\end{figure}
%-----------------------------------------------------------------------

%-----------------------------------------------------------------------
\begin{lemma}[Minkowski--Cauchy Formula] \label{lem:minkowski-cauchy}
Let $K$ be a convex body in $\RE^d$ containing the origin in its interior. For any convex body $G$, its Holmes--Thompson surface area with respect to the gauge $\|\cdot\|_K$ is given by
\[
    \areaMink[K](G) 
        ~ = ~ 
        \frac{1}{\omega_{d-1}} \int_{S^{d-1}} \lambda_{d-1}\left( S^M_K(G, u) \right) \, d\sigma(u).
\]
\end{lemma}
%-----------------------------------------------------------------------

%-----------------------------------------------------------------------
\begin{proof}
We derive this result as the limit of the Funk--Cauchy formula (Theorem~\ref{thm:cauchy-funk-gen}). Consider the scaled bodies $K_r := r K$. For each $x \in \bd G$, the Funk surface-area density associated with $K_r$ is determined by $\polar{(K_r-x)} = r^{-1}\polar{(K-x/r)}$. Since $x/r \to O$ uniformly on the compact set $\bd G$, continuity of polarity and orthogonal projection for convex bodies implies that the corresponding density is $r^{-(d-1)}$ times the Minkowski density, up to a factor of $1+o(1)$ uniform on $\bd G$. Consequently,
\[
    \areaFunk[K_r](G)
        ~ = ~ r^{-(d-1)} \cdot \areaMink[K](G)\,(1+o(1)),
\]
where $o(1)$ denotes a scalar tending to $0$ as $r \to \infty$. 

Fix $u\in S^{d-1}$ and any $r$ that is sufficiently large that $G \subset \interior(K_r)$. For the sake of illustration, let us think of $u$ as directed vertically upwards (see Figure~\ref{f:minkowski-shadows-2}(a)). On the right-hand side of Theorem~\ref{thm:cauchy-funk-gen}, the central shadow $S_{K_r}(G,u)$ lies on the hyperplane $-u^* = \{y \ST \inner{y}{u} = -1\}$. To simplify notation, let
\[
    h    ~ := ~ h_K(u) \quad\text{and}\quad
    v    ~ := ~ v_K(u).
\]
Since $v$ is a support point, we have $\inner{v}{u} = h$. We can identify points of $-u^*$ with $u^\perp$ through parallel projection along $v$ by defining $\tau(y) := y + (1/h) v$.

%-----------------------------------------------------------------------
\begin{figure}[htbp]
  \centerline{\includegraphics[scale=0.40,page=2]{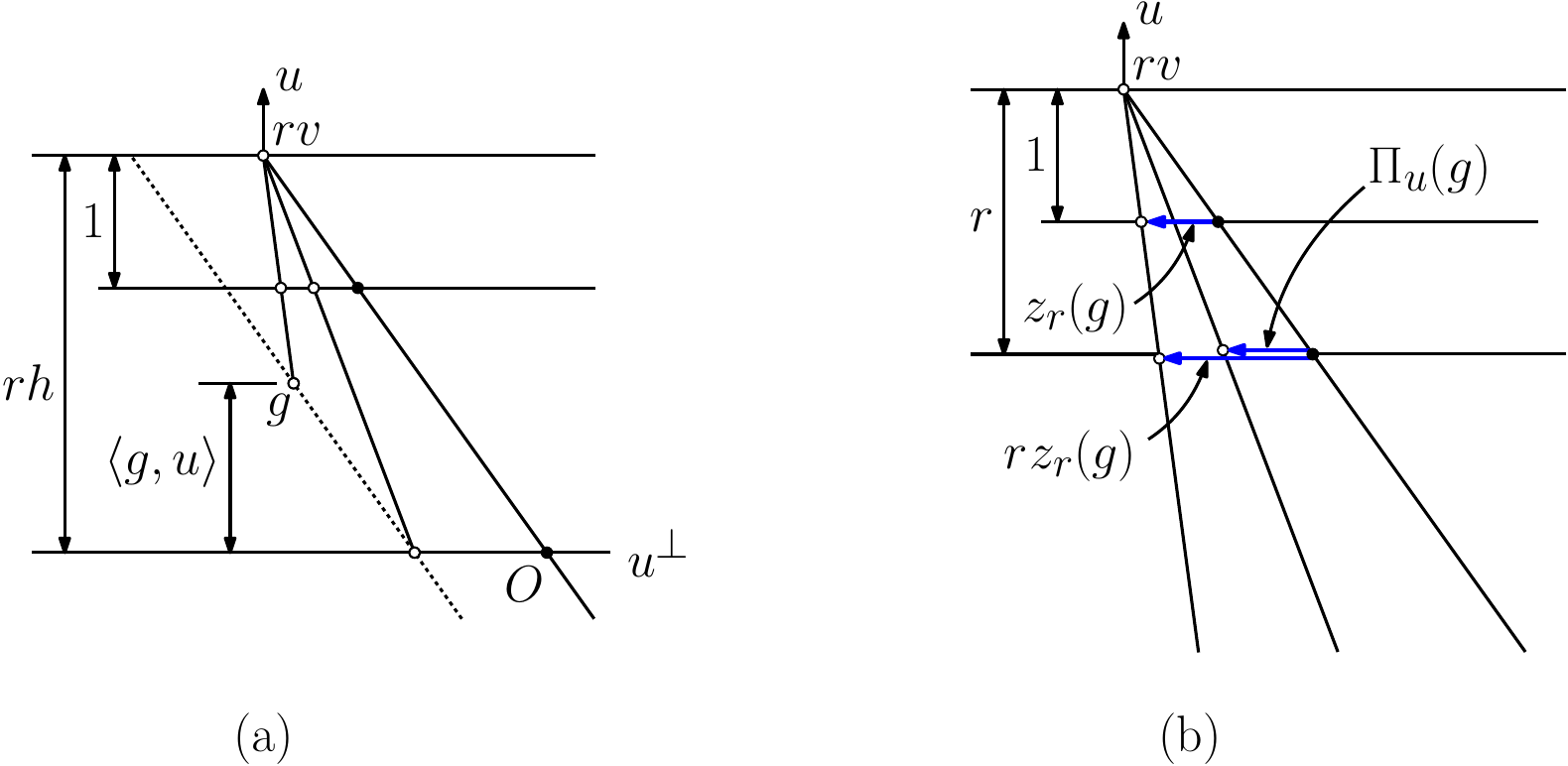}}
  \caption{Proof of Lemma~\ref{lem:minkowski-cauchy}. (The free vectors $z_r(g)$, $r \kern+1pt z_r(g)$, and $\Pi_u(g)$, which lie on $u^\perp$, have been translated to better illustrate their lengths.)}
  \label{f:minkowski-shadows-2}
\end{figure}
%-----------------------------------------------------------------------

For $g \in G$, consider where the ray from $r v$ through $g$ intersects a horizontal hyperplane that lies one unit below $r v$. To relate this to our definition of shadows, let us translate by $-r v$, and consider the ray from the origin through $g - r v$, i.e., $\{t(g - r v) \ST t \geq 0\}$ instead. It is easily verified that its intersection with $-u^*$ occurs at the point
\[
    y_r(g) ~ := ~ t_r(g) (g - r v),
    \quad\text{where}~~
    t_r(g) ~ := ~ \frac{1}{r h - \inner{g}{u}}.
\]
Let $z_r(g) := \tau(y_r(g))$ be the associated point on $u^\perp$. (Figure~\ref{f:minkowski-shadows-2}(b) illustrates $z_r(g)$ as a vector lying on the horizontal hyperplane one unit below $r v$.)

A direct expansion using the boundedness of $G$ and $K$ yields
\begin{equation}
    r \kern+1pt z_r(g)
        ~ = ~ \frac{1}{h} \left( g - \frac{\inner{g}{u}}{h} v \right) + O\left( \frac{1}{r} \right), \label{eq:minkowski-cauchy-1}
\end{equation}
where the $O(1/r)$ term is uniform for all $g \in G$ and $u \in S^{d-1}$. (Figure~\ref{f:minkowski-shadows-2}(b) illustrates these two vectors as lying on the horizontal hyperplane $r$ units below $r v$.) Observe that the leading term of Eq.~\eqref{eq:minkowski-cauchy-1} is exactly $\Pi_u(g)$, hence the convex bodies
\[
    A_r(u)
        ~ := ~ r \left(\tau(S_{K_r}(G,u))\right)
        ~ \subseteq ~ u^\perp
    \quad\text{and}\quad
    S^{\mathrm M}_K(G,u)
        ~ =~ \Pi_u(G)
\]
converge in Hausdorff distance at rate $O(1/r)$. Since volume is continuous under Hausdorff convergence for convex bodies, we obtain
\[
    \lambda_{d-1}(A_r(u)) 
        ~ \to ~ \lambda_{d-1} \left(S^{\mathrm M}_K(G,u)\right).
\]

Finally, $\tau$ and translations preserve $(d-1)$-volumes on the hyperplane, and scaling by $r$ multiplies $(d-1)$-volume by $r^{d-1}$, so
\[
    \lambda_{d-1}(A_r(u))
        ~ = ~ r^{d-1} \lambda_{d-1}(S_{K_r}(G,u)).
\]
Substituting into the Funk--Cauchy formula, multiplying by $r^{d-1}$, and letting $r \to \infty$ yields
\[
    \areaMink[K](G)
        ~ = ~ \frac{1}{\omega_{d-1}} \int_{S^{d-1}} \lambda_{d-1} \left(S^{\mathrm M}_K(G,u)\right) \, d\sigma(u),
\]
as desired.
\end{proof}
%-----------------------------------------------------------------------

Observe that in the special case where $K$ is the Euclidean unit ball $B^d$, the Minkowski shadow $S_K^M(G, u)$ becomes the standard orthogonal projection, and we recover the classical Euclidean Cauchy formula (Eq.~\eqref{eq:cauchy-euclidean}).

Holmes and Thompson (see \cite[Corollary 2.13]{HoT79}) also give a Cauchy-type formula for the surface area in Minkowski geometry. They average Euclidean orthogonal projection areas, with the parameter ranging over the boundary of the polar unit ball, $\bd \polar{K}$. By contrast, Lemma~\ref{lem:minkowski-cauchy} averages the area of the Minkowski shadow $S_K^M(G, u)$ over directions $u \in S^{d-1}$, and it arises as the scaling limit of the Funk Cauchy formula based on central shadows. From a computational perspective, our formula has a simpler sampling distribution and does not require computing polar bodies---only support points, projections, and scalings are needed.

%=======================================================================
\section*{Concluding Remarks and Open Problems}
%=======================================================================

We established an explicit Cauchy-type formula for the Holmes--Thompson surface area in the Funk geometry induced by a convex body $K \subset \RE^d$. This formula expresses the surface area of a convex body $G$ as the average, over directions $u\in S^{d-1}$, of the $(d-1)$-dimensional measures of the corresponding \emph{central shadows} of $G$ (Theorem~\ref{thm:cauchy-funk-gen}). For polytopal $K$, this identity admits a discrete vertex-based decomposition (Theorem~\ref{thm:cauchy-funk-polytope}), and the same framework yields a Crofton-type representation in terms of an explicit measure on the space of unoriented lines intersecting $G$ (Theorem~\ref{thm:crofton-unoriented}). 

Taken together, these results provide a concrete shadow-averaging principle for surface area in a projective Finsler setting. They also show that several classical formulas from Euclidean, Minkowski, Hilbert, and hyperbolic geometries can be understood within a single framework. From a computational perspective, our formulas involve Euclidean $(d-1)$-dimensional volumes of explicitly defined slices or projections, thereby avoiding direct evaluation of the Holmes--Thompson definitions through polars of pointwise Finsler balls. In the polytopal case, the resulting decomposition is especially simple, replacing more combinatorial Crofton descriptions based on pairs of faces by a sum over vertices and normal cones. 

Our work raises several natural open problems. A first direction is to develop provably efficient randomized algorithms for estimating Funk surface area from the Cauchy and Crofton formulas, ideally with explicit variance bounds and high-probability guarantees. For polytopal $K$, this includes efficient sampling from the spherical normal-cone decomposition and efficient evaluation or approximation of the associated central shadows. It would also be interesting to determine to what extent these formulas can serve as practical primitives for geometric computation in Funk- and Hilbert-type domains, in a manner analogous to the role of Cauchy and Crofton formulas in Euclidean stereology, tomography, and randomized surface area estimation. 

A second direction is to investigate whether this shadow-based approach extends beyond surface area to higher-order intrinsic quantities, such as the Holmes--Thompson analogs of quermassintegrals or curvature measures. A third is to better understand the relationship with Hilbert geometry, especially by identifying conditions on $K$ under which the Funk formula yields the exact Hilbert surface area beyond the known cases of $d=2$ and ellipsoids. More generally, it would be interesting to determine how geometric properties of $K$ govern the approximation gap between Hilbert and Funk surface areas.

%=======================================================================
% Bibliography
%=======================================================================

\pdfbookmark[1]{References}{s:ref}

%-----------------------------------------------------------------------
% Appendices
%-----------------------------------------------------------------------

\appendix

%=======================================================================
\section{Justification of the Cone Volume Definition} \label{s:just-cone-vol}
%=======================================================================

The definition of the Funk volume for cones (Eq.~\eqref{eq:funk-vol-cone-def}) relies on an integration over the cone's spherical cross-section. To justify this definition, we must show that it is consistent with the standard Holmes--Thompson volume defined on cross-sections. Specifically, the value should be invariant to the choice of the slicing hyperplane, subject to the condition that the intersection is bounded. 

We begin by establishing a generalized duality statement in Lemma~\ref{lem:slice-gen}. While the duality between sections of a convex body and projections of its polar is well-known (see Lemma~\ref{lem:slice}), our treatment of cones requires a corresponding result for possibly unbounded closed convex sets. To this end, we adapt the argument of~\cite[Theorem 2.2.9 and Corollary 2.2.10]{Tho96}. We begin with the following technical identity. Recall from Section~\ref{s:prelim} that, for a closed convex set $K$ containing the origin and a linear subspace $E$, the notation $\polarIn[E]{K}$ denotes the polar of $K$ taken within $E$.

%-----------------------------------------------------------------------
\begin{lemma} \label{lem:slice-aux}
For any set $K \subseteq \RE^d$ and any linear subspace $E \subset \RE^d$,
\[
    \polar{K} \cap E 
        ~ = ~ \polarIn[E]{(\orthproj{K}{E})}.
\]
\end{lemma}
%-----------------------------------------------------------------------

%-----------------------------------------------------------------------
\begin{proof}
Let $u \in E$. Since the projection is orthogonal, for any $x \in K$ we have $\inner{x}{u} = \inner{\big(\orthproj{x}{E}\big)}{u}$. Consequently, the support functions coincide on $E$:
\[
    h_K(u) 
        ~ = ~ \sup_{x \in K} \inner{x}{u} 
        ~ = ~ \sup_{x \in K} \inner{\orthproj{x}{E}}{u} 
        ~ = ~ h_{\orthprojsub{K}{E}}(u).
\]
By definition, a vector $u$ belongs to $\polar{K} \cap E$ if and only if $u \in E$ and $h_K(u) \leq 1$. Using the identity above, for $u \in E$, this is equivalent to $h_{\orthprojsub{K}{E}}(u) \leq 1$. This is precisely the definition of $u$ belonging to $\polarIn[E]{(\orthproj{K}{E})}$.
\end{proof}
%-----------------------------------------------------------------------

%-----------------------------------------------------------------------
\begin{lemma} \label{lem:slice-gen}
Let $K \subseteq \RE^d$ be a (possibly unbounded) closed convex set with the origin in its interior, and let $E$ be a linear subspace of $\RE^d$. Then the polar of the section $K \cap E$, taken relative to the subspace $E$, is the projection of the polar of $K$ onto $E$. That is,
\[
    \polarIn[E]{(K \cap E)} 
        ~ = ~ \orthproj{\polar{K}}{E}.
\]
\end{lemma}
%-----------------------------------------------------------------------

%-----------------------------------------------------------------------
\begin{proof}
Applying Lemma~\ref{lem:slice-aux} to $\polar{K}$ gives
\[
    \polar{(\polar{K})} \cap E 
        ~ = ~ \polarIn[E]{(\orthproj{\polar{K}}{E})}.
\] 
Since $K$ is a closed convex set containing the origin, the Bipolar Theorem~\cite[Theorem 14.5]{Roc97} implies that  $\polar{(\polar{K})} = K$. Thus,
\[
    K \cap E 
        ~ = ~ \polarIn[E]{(\orthproj{\polar{K}}{E})}.
\]
We now take the polar in $E$ on both sides. For the left side, we obtain $\polarIn[E]{(K \cap E)}$. For the right side, let $G = \orthproj{\polar{K}}{E}$. Since $O \in \interior(K)$, the polar $\polar{K}$ is compact. The orthogonal projection of a compact set is compact, so $G$ is a compact convex subset of $E$ containing the origin. Applying the Bipolar Theorem within the subspace $E$, we have $\polarIn[E]{(\polarIn[E]{G})} = G$. Therefore  $\polarIn[E]{(K \cap E)} = \orthproj{\polar{K}}{E}$.
\end{proof}
%-----------------------------------------------------------------------

Recall from Section~\ref{s:cones} that, given a pointed cone $K$ and a point $x$ in $K$'s interior, the dual cross-section, $F_K(x)$, is defined to be the intersection of the polar cone $\polar{K}$ with the dual hyperplane associated with $-x$, that is, $F_K(x) = \polar{K} \cap -x^*$.

%-----------------------------------------------------------------------
\begin{lemma}\label{lem:cone-aux}
Let $K \subset\RE^d$ be a pointed cone and $x \in \interior(K)$. Let $H$ be any hyperplane such that the intersection $K \cap H$ is bounded. Then $\orthproj{\polar{(K-x)}}{E} = \orthproj{F_K(x)}{E}$, where $E$ is the linear subspace parallel to $H$.
\end{lemma}
%-----------------------------------------------------------------------

%-----------------------------------------------------------------------
\begin{proof}
Recall from Eq.~\eqref{eq:F-def} that $F_K(x)$ is the intersection of $\polar{K}$ with the hyperplane $-x^*$. Since $x \in \interior(K)$, $F_K(x)$ is bounded and $\polar{(K-x)} = \conv(O,F_K(x))$. Let $u$ be a normal vector to $H$. The boundedness of $K \cap H$ implies that the line spanned by $u$ passes through $\interior(\polar{K})$, and therefore intersects $F_K(x)$. Thus, the origin is contained in $\orthproj{F_K(x)}{E}$, which implies that the projection of the convex hull reduces to the projection of $F_K(x)$.
\end{proof}
%-----------------------------------------------------------------------

We now present the main result of this section. Recall that, given a cone $K$, a hyperplane $H$ is \emph{admissible} for $K$ if $K \cap H$ is bounded.

%-----------------------------------------------------------------------
\begin{lemma} \label{lem:cone-vol}
Let $K, G \subset \RE^d$ be two pointed cones, such that $(G \setminus \{O\}) \subset \interior(K)$. Let $H$ be any hyperplane that is admissible for $K$. Then 
\[
    \volFunk[K](G) ~ = ~ \volFunk[K \cap H](G \cap H).
\]
\end{lemma}
%-----------------------------------------------------------------------

%-----------------------------------------------------------------------
\begin{proof}
Recall by definition, 
\[
    \volFunk[K](G)
    ~ = ~ \frac{1}{\omega_{d-1}} \int_{\Sigma(G)} \lambda_{d-1}(F_K(s)) \, d\sigma(s),                  \quad\text{where $\Sigma(G) = S^{d-1} \cap G$.}
\]
Let $H$ be an admissible hyperplane with unit normal $u$, and let $E$ be its parallel subspace (see Figure~\ref{f:cone-5}(a)). Let $K_H = K \cap H$ and $G_H = G \cap H$. By the definition of $\volFunk[K_H](G_H)$ and Lemmas~\ref{lem:slice-gen} and~\ref{lem:cone-aux}, we have
\begin{align}
    \volFunk[K_H](G_H) 
    & ~ = ~ \frac{1}{\omega_{d-1}} \int_{G_H} \lambda_{d-1}(\polarIn[E]{(K_H-x)}) \, d\lambda_{d-1}(x) \nonumber\\
    & ~ = ~ \frac{1}{\omega_{d-1}} \int_{G_H} \lambda_{d-1}(\orthproj{F_K(x)}{E}) \, d\lambda_{d-1}(x). \label{eq:def-vol-section}
\end{align}

%-----------------------------------------------------------------------
\begin{figure}[htbp]
  \centerline{\includegraphics[scale=0.40,page=5]{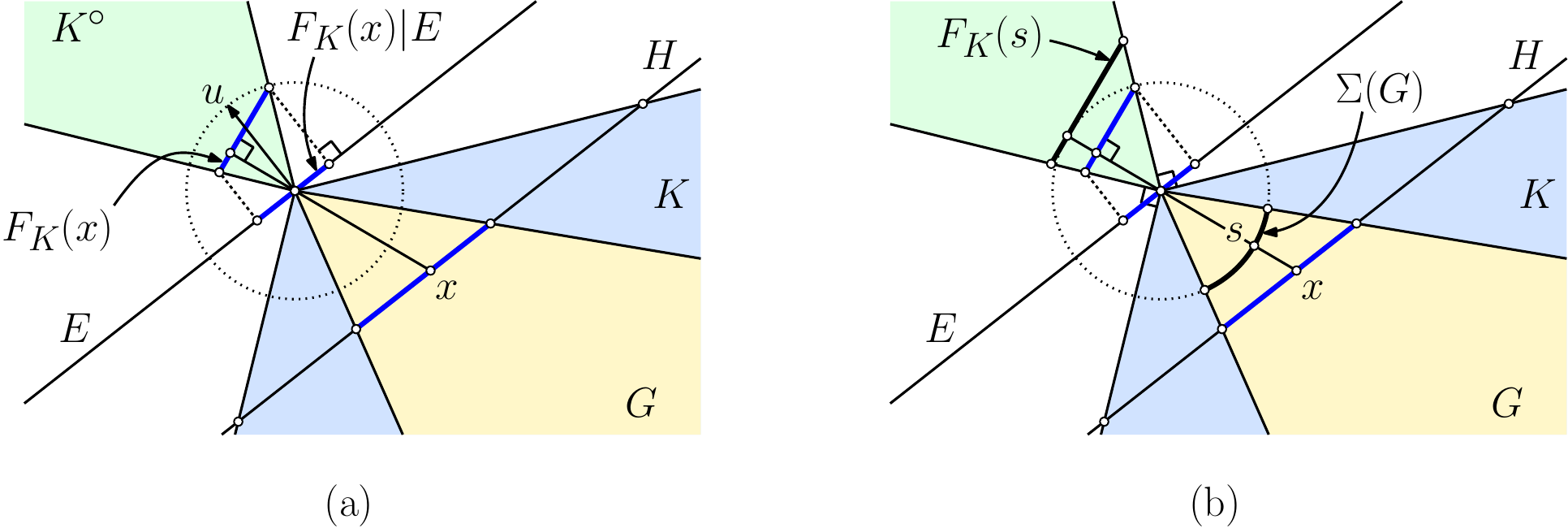}}
  \caption{Proof of Lemma~\ref{lem:cone-vol}.}
  \label{f:cone-5}
\end{figure}
%-----------------------------------------------------------------------

Consider the radial map $x \mapsto s := x / \|x\|$ from $G_H$ to $\Sigma(G)$ (see Figure~\ref{f:cone-5}(b)). Since $K$ is pointed and $O \notin H$, this is a bijection. The transformation for surface measures is
\begin{equation} \label{eq:measure-transform-H}
    d\sigma(s)
        ~ = ~ \frac{|\inner{s}{u}|}{\|x\|^{d-1}} \, d\lambda_{d-1}(x).
\end{equation}
As in the proof of Lemma~\ref{lem:cone-boundary-integral}, homogeneity implies that
\begin{equation} \label{eq:scaling-H}
    \lambda_{d-1}(F_K(s)) 
        ~ = ~ \|x\|^{d-1} \, \lambda_{d-1}(F_K(x)).
\end{equation}
Further, $F_K(x)$ lies on a hyperplane whose unit normal is $s$, and $E$ has unit normal $u$, so 
\begin{equation} \label{eq:projection-H}
    \lambda_{d-1} (\orthproj{F_K(x)}{E}) 
        ~ = ~ |\!\inner{s}{u}\!| \, \lambda_{d-1}(F_K(x)).
\end{equation}
Combining Eqs.~\eqref{eq:measure-transform-H}--\eqref{eq:projection-H} gives
\[
    \lambda_{d-1}(\orthproj{F_K(x)}{E}) \, d \lambda_{d-1}(x) 
        ~ = ~ \lambda_{d-1}(F_K(s)) \, d \sigma(s).
\]
Since the radial map is a bijection, integrating both sides yields the result.
\end{proof}
%-----------------------------------------------------------------------

%=======================================================================
\section{A Direct Geometric Proof of Funk Area Duality} \label{s:polarity-conservation}
%=======================================================================

In this appendix, we provide a geometric proof of the polarity invariance of Funk surface area. This duality was established by Faifman~\cite{Fai24} by symplectic methods. An alternative elementary proof was given in Appendix~C of~\cite{ArM26Darxiv}. Unlike that earlier elementary proof, which establishes the polarity identity directly, the proof presented here derives it from the vertex decomposition of Section~\ref{s:vertex-decomp-ptope} together with Faifman's duality for Funk volume on convex bodies. In the polytopal case, this approach also yields a facet-based realization of the duality.

%-----------------------------------------------------------------------
\begin{lemma}[Funk Area Duality] \label{lem:area-duality}
Let $K$ be a convex body in $\RE^d$ containing the origin in its interior, and let $G \subset \interior(K)$ be a convex body with $O \in \interior(G)$. Then
\[
    \areaFunk[K](G) ~ = ~ \areaFunk[\polar{G}](\polar{K}).
\]
\end{lemma}
%-----------------------------------------------------------------------

The proof is based on the following cone-volume duality statement.

%-----------------------------------------------------------------------
\begin{lemma}[Cone Duality] \label{lem:cone-duality}
Let $K \subset \RE^d$ be a pointed cone with apex at the origin, and let $G \subset \interior(K)$ be a pointed subcone. Then
\[
    \volFunk[K](G) ~ = ~ \volFunk[\polar{G}](\polar{K}).
\]
\end{lemma}
%-----------------------------------------------------------------------

%-----------------------------------------------------------------------
\begin{proof} 
After an invertible linear change of coordinates, we may assume that $K, G \subset \{x \in \RE^d \ST x_d < 0\}$, and further, $-e_d \in \interior(G)$, where $e_d$ denotes the $d$th coordinate unit vector. Set 
\[
    H_K ~ := ~ \{x' \in \RE^{d-1} \ST (x',-1) \in K\}
        \quad\text{and}\quad
    H_G ~ := ~ \{x' \in \RE^{d-1} \ST (x',-1) \in G\}
\]
(see Figure~\ref{f:cone-duality}).

%-----------------------------------------------------------------------
\begin{figure}[htbp]
  \centerline{\includegraphics[scale=0.40]{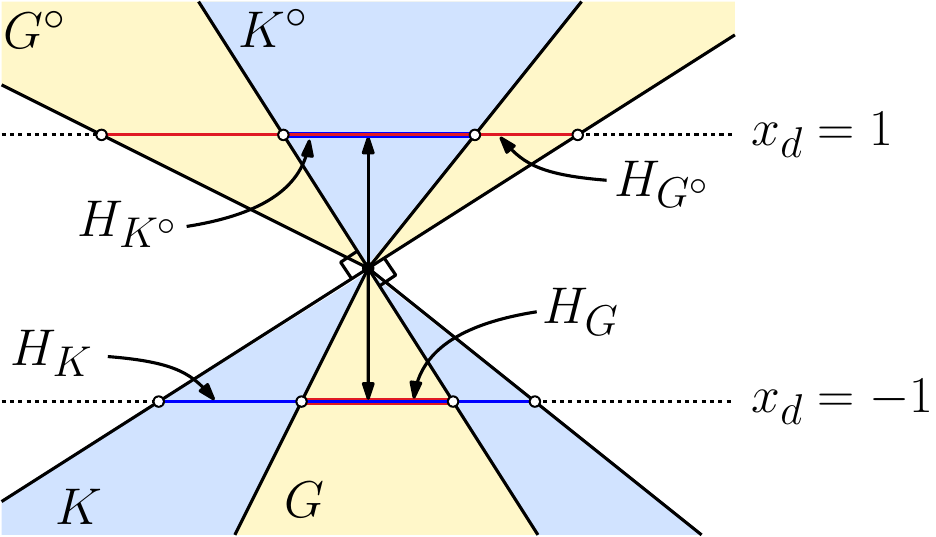}}
  \caption{Proof of Lemma~\ref{lem:cone-duality}. (The bodies $H_K$, $H_G$, $H_{\polar{K}}$, and $H_{\polar G}$ have been translated to hyperplanes $x_d = 1$ and $x_d = -1$.)}
  \label{f:cone-duality}
\end{figure}
%-----------------------------------------------------------------------

Clearly, $H_K$ and $H_G$ are convex bodies in $\RE^{d-1}$ with $H_G \subset \interior(H_K)$. By Lemma~\ref{lem:cone-vol}, $\volFunk[K](G) = \volFunk[H_K](H_G)$. Similarly, define 
\[
    H_{\polar K} ~ := ~ \{y' \in \RE^{d-1} \ST (y',1) \in \polar{K}\}
    \qquad\text{and}\qquad
    H_{\polar G} ~ := ~ \{y' \in \RE^{d-1} \ST (y',1) \in \polar{G}\}.
\]
For any $y' \in \RE^{d-1}$, we have
\begin{align*}
    y' \in H_{\polar K}
            & \iff \inner{(x',-1)}{(y',1)} \leq 0, \quad \text{for all $x' \in H_K$} \\
            & \iff \inner{x'}{y'} \leq 1, \quad \text{for all $x' \in H_K$} \\
            & \iff y' \in \polar{H_K}.
\end{align*}
Hence $H_{\polar K} = \polar{H_K}$, and similarly $H_{\polar G} = \polar{H_G}$. Applying Faifman's polarity duality for Funk volume on convex bodies~\cite{Fai24}, we obtain 
\[
    \volFunk[H_K](H_G)
        ~ = ~ \volFunk[\polar{H_G}](\polar{H_K})
        ~ = ~ \volFunk[H_{\polar G}](H_{\polar K}).
\]
A second application of Lemma~\ref{lem:cone-vol} gives 
\[
    \volFunk[H_{\polar G}](H_{\polar K})
        ~ = ~ \volFunk[\polar{G}](\polar{K}),
\]
and the proof follows.
\end{proof}
%-----------------------------------------------------------------------

%-----------------------------------------------------------------------
\begin{proof}[Proof of Lemma~\ref{lem:area-duality}] 
We first consider the case where $K$ is a convex polytope. Recall that $V(K)$ denotes the vertex set of $K$. By Theorem~\ref{thm:cauchy-funk-polytope}, we have
\[
    \areaFunk[K](G)
        ~ = ~ \sum\nolimits_{v \in V(K)} \volFunk[K_v](G_v),
\]
where $K_v := \cone(K-v)$ and $G_v := \cone(G-v)$. For each $v \in V(K)$, Lemma~\ref{lem:cone-duality} gives 
\[
    \volFunk[K_v](G_v)
        ~ = ~ \volFunk[\polar{G_v}](\polar{K_v}).
\]
Fix $v \in V(K)$ (see Figure~\ref{f:area-duality}(a)). Slice both cones $\polar{K_v}$ and $\polar{G_v}$ by the dual hyperplane $v^*$ defined in Eq.~\eqref{eq:dual-hyperplane} (see Figure~\ref{f:area-duality}(b)). Recall that $F_v$ denotes the facet of $\polar{K}$ dual to $v$. By Corollary~\ref{cor:polar-facets}, we have $\polar{K_v} \cap v^* = F_v$.

%-----------------------------------------------------------------------
\begin{figure}[htbp]
  \centerline{\includegraphics[scale=0.40]{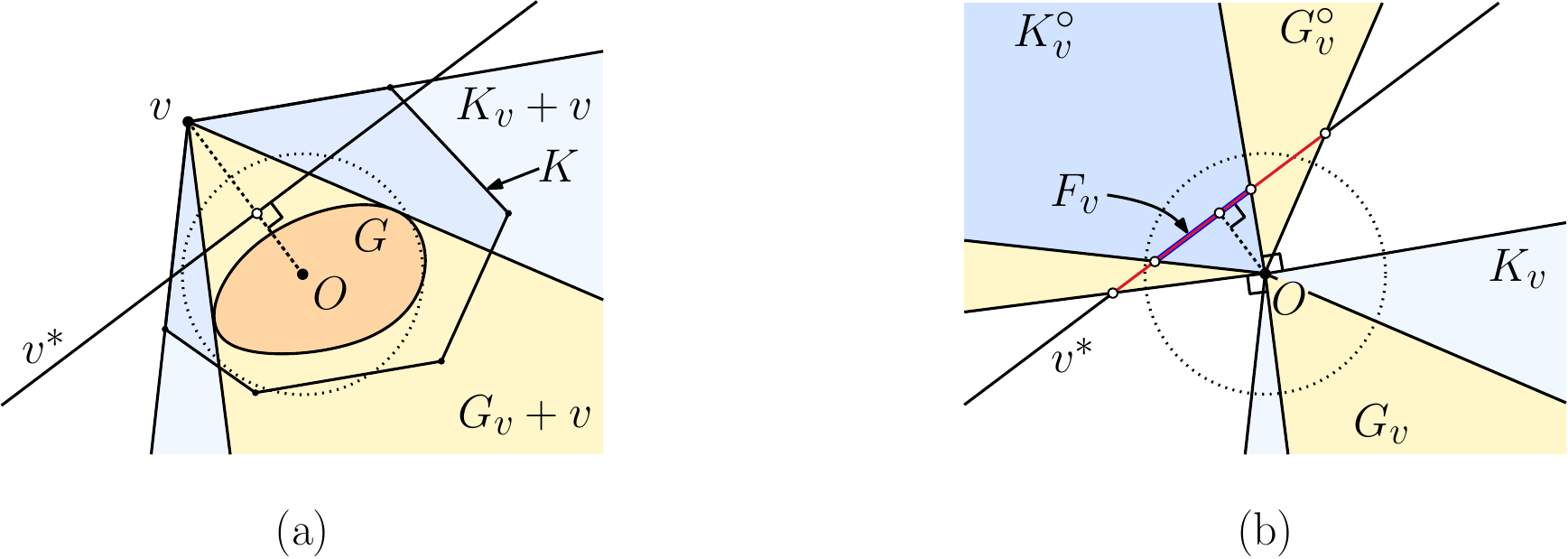}}
  \caption{Proof of Lemma~\ref{lem:area-duality}.}
  \label{f:area-duality}
\end{figure}
%-----------------------------------------------------------------------s

Since $v \notin \interior(G)$, Lemma~\ref{lem:polar-slice} applied to $Q=G$ yields
\[
    \polar{G_v} \cap v^* ~ = ~ \polar{G} \cap v^*.
\]
Because $\polar{G}$ is compact, the section $\polar{G}\cap v^*$ is bounded, so $v^*$ is admissible for $\polar{G_v}$. Therefore, by Lemma~\ref{lem:cone-vol}, 
\[
    \volFunk[\polar{G_v}](\polar{K_v})
        ~ = ~ \volFunk[\polar{G} \cap v^*](F_v).
\]
Identifying the hyperplane $v^*$ with the linear space $v^\perp$ by translation, the right-hand side is exactly the Funk surface-area contribution of the facet $F_v$ to $\areaFunk[\polar{G}](\polar{K})$. Thus,
\[
    \volFunk[\polar{G}\cap v^*](F_v)
        ~ = ~ \areaFunk[\polar{G}](F_v).
\]
Combining the preceding identities, we obtain $\volFunk[K_v](G_v) = \areaFunk[\polar{G}](F_v)$. Summing over all $v \in V(K)$, and using that the facets $F_v$ cover $\bd \polar{K}$ up to lower-dimensional overlaps, we obtain
\[
    \areaFunk[K](G)
        ~ = ~ \sum_{v \in V(K)} \areaFunk[\polar{G}](F_v)
        ~ = ~ \areaFunk[\polar{G}](\polar{K}).
\]
This proves the lemma in the polytopal case. 

The general case follows by choosing a sequence of polytopes $K_n \to K$ in the Hausdorff metric. Since $G \subset \interior K$, after discarding finitely many terms we may assume that $G \subset \interior(K_n)$ for all $n$, and then pass to the limit using continuity of polarity and of the Holmes--Thompson area functional.
\end{proof} 
%-----------------------------------------------------------------------


\begin{thebibliography}{10}

\bibitem{ACF22}
C.~Aaron, A.~Cholaquidis, and R.~Fraiman.
\newblock Estimation of surface area.
\newblock {\em Electronic Journal of Statistics}, 16(2):3751--3788, 2022.
\newblock \href {https://doi.org/10.1214/22-EJS2031} {\path{doi:10.1214/22-EJS2031}}.

\bibitem{AbM24}
A.~Abdelkader and D.~M. Mount.
\newblock Convex approximation and the {Hilbert} geometry.
\newblock In {\em Proc.\ SIAM Symp.\ Simplicity in Algorithms (SOSA24)}, pages 286--298, 2024.
\newblock \href {https://doi.org/10.1137/1.9781611977936.26} {\path{doi:10.1137/1.9781611977936.26}}.

\bibitem{AHV04}
P.~K. Agarwal, S.~Har-Peled, and K.~R. Varadarajan.
\newblock Approximating extent measures of points.
\newblock {\em J.\ Assoc.\ Comput.\ Mach.}, 51:606--635, 2004.
\newblock \href {https://doi.org/10.1145/1008731.1008736} {\path{doi:10.1145/1008731.1008736}}.

\bibitem{Ale78}
R.~Alexander.
\newblock Planes for which the lines are the shortest paths between points.
\newblock {\em Ill.\ J.\ Math.}, 22:177--190, 1978.
\newblock \href {https://doi.org/10.1215/ijm/1256048729} {\path{doi:10.1215/ijm/1256048729}}.

\bibitem{ABF05}
R.~Alexander, I.~D. Berg, and R.~L. Foote.
\newblock Integral-geometric formulas for perimeter in {$S^2$}, {$H^2$}, and {Hilbert} planes.
\newblock {\em Rocky Mountain J.\ Mathematics}, 35(6):1825--1859, 2005.
\newblock URL: \url{http://www.jstor.org/stable/44238783}.

\bibitem{AlF98}
J.~C. {\'{A}lvarez Paiva} and E.~Fernandes.
\newblock Crofton formulas in projective {Finsler} spaces.
\newblock {\em Electron.\ Res.\ Announc.\ Amer.\ Math.\ Soc.}, 4:91--100, 1998.
\newblock URL: \url{http://eudml.org/doc/224916}.

\bibitem{AAFM22}
R.~Arya, S.~Arya, G.~D. da~Fonseca, and D.~M. Mount.
\newblock Optimal bound on the combinatorial complexity of approximating polytopes.
\newblock {\em ACM Trans.\ Algorithms}, 18:1--29, 2022.
\newblock \href {https://doi.org/10.1145/3559106} {\path{doi:10.1145/3559106}}.

\bibitem{AFM24}
S.~Arya, G.~D. da~Fonseca, and D.~M. Mount.
\newblock Economical convex coverings and applications.
\newblock {\em SIAM J.\ Comput.}, 53(4):1002--1038, 2024.
\newblock \href {https://doi.org/10.1137/23M1568351} {\path{doi:10.1137/23M1568351}}.

\bibitem{ArM26Darxiv}
S.~Arya and D.~M. Mount.
\newblock On the duality of coverings in hilbert geometry, 2026.
\newblock URL: \url{https://arxiv.org/abs/2603.18929}, \href {https://arxiv.org/abs/2603.18929} {\path{arXiv:2603.18929}}.

\bibitem{BaJ05}
A.~Baddeley and E.~B.~V. Jensen.
\newblock {\em Stereology for Statisticians}.
\newblock CRC Press, 2005.
\newblock \href {https://doi.org/10.1201/9780203496817} {\path{doi:10.1201/9780203496817}}.

\bibitem{Bar02}
A.~Barvinok.
\newblock {\em A Course in Convexity}, volume~54 of {\em Graduate Studies in Mathematics}.
\newblock American Mathematical Society, 2002.
\newblock URL: \url{https://bookstore.ams.org/gsm-54/}.

\bibitem{BNN10}
J.-D. Boissonnat, F.~Nielsen, and R.~Nock.
\newblock Bregman {Voronoi} diagrams.
\newblock {\em Discrete Comput.\ Geom.}, 44:281--307, 2010.
\newblock \href {https://doi.org/10.1007/s00454-010-9256-1} {\path{doi:10.1007/s00454-010-9256-1}}.

\bibitem{CFK97}
J.~W. Cannon, W.~J. Floyd, R.~Kenyon, and W.~R. Parry.
\newblock Hyperbolic geometry.
\newblock In S.~Levy, editor, {\em Flavors of Geometry}, volume~31 of {\em MSRI Publications}, pages 59--115. Cambridge University Press, 1997.

\bibitem{Cla06}
K.~L. Clarkson.
\newblock Building triangulations using $\varepsilon$-nets.
\newblock In {\em Proc.\ 38th Annu.\ ACM Sympos.\ Theory Comput.}, pages 326--335, 2006.
\newblock \href {https://doi.org/10.1145/1132516.1132564} {\path{doi:10.1145/1132516.1132564}}.

\bibitem{Car76}
M.~P. do~Carmo.
\newblock {\em Differential Geometry of Curves and Surfaces}.
\newblock Prentice Hall, 1976.

\bibitem{Fai24}
D.~Faifman.
\newblock A {Funk} perspective on billiards, projective geometry and {Mahler} volume.
\newblock {\em J.\ Differential Geom.}, 127:161--212, 2024.
\newblock \href {https://doi.org/10.4310/jdg/1717356157} {\path{doi:10.4310/jdg/1717356157}}.

\bibitem{FVW23}
D.~Faifman, C.~Vernicos, and C.~Walsh.
\newblock Volume growth of {Funk} geometry and the flags of polytopes.
\newblock {\em Geometry \& Topology}, 29(7):3773--3811, 2025.
\newblock \href {https://doi.org/10.2140/gt.2025.29.3773} {\path{doi:10.2140/gt.2025.29.3773}}.

\bibitem{Gar06}
R.~J. Gardner.
\newblock {\em Geometric Tomography}.
\newblock Cambridge University Press, 2nd edition, 2006.
\newblock \href {https://doi.org/10.1017/CBO9781107341029} {\path{doi:10.1017/CBO9781107341029}}.

\bibitem{HoT79}
R.~D. Holmes and A.~C. Thompson.
\newblock $n$-dimensional area and content in {Minkowski} spaces.
\newblock {\em Pacific J.\ Math.}, 85:77--110, 1979.
\newblock \href {https://doi.org/10.2140/pjm.1979.85.77} {\path{doi:10.2140/pjm.1979.85.77}}.

\bibitem{ISY11}
M.~Ito, Y.~Seo, T.~Yamazaki, and M.~Yanagida.
\newblock Geometric properties of positive definite matrices cone with respect to the {Thompson} metric.
\newblock {\em Linear Algebra Appl.}, 435:2054--2064, 2011.
\newblock \href {https://doi.org/10.1016/j.laa.2011.03.063} {\path{doi:10.1016/j.laa.2011.03.063}}.

\bibitem{KPK10}
D.~Krioukov, F.~Papadopoulos, M.~Kitsak, A.~Vahdat, and M.~Bogu{\~ n}{\' a}.
\newblock Hyperbolic geometry of complex networks.
\newblock {\em Phys. Rev. E}, 82, 2010.
\newblock \href {https://doi.org/10.1103/PhysRevE.82.036106} {\path{doi:10.1103/PhysRevE.82.036106}}.

\bibitem{KuV06}
A.~A. Kurzhanskiy and P.~Varaiya.
\newblock Ellipsoidal toolbox.
\newblock Technical Report UCB/EECS-2006-46, EECS Department, University of California, Berkeley, 2006.
\newblock URL: \url{https://www2.eecs.berkeley.edu/Pubs/TechRpts/2006/EECS-2006-46.html}.

\bibitem{LWM03}
X.~Li, W.~Wang, R.~R. Martin, and A.~Bowyer.
\newblock Using low-discrepancy sequences and the crofton formula to compute surface areas of geometric models.
\newblock {\em Computer-Aided Design}, 35(9):771--782, 2003.
\newblock \href {https://doi.org/10.1016/S0010-4485(02)00100-8} {\path{doi:10.1016/S0010-4485(02)00100-8}}.

\bibitem{LDL23}
W.~Lin, V.~Duruisseaux, M.~Leok, F.~Nielsen, M.~E. Khan, and M.~Schmidt.
\newblock Simplifying momentum-based positive-definite submanifold optimization with applications to deep learning.
\newblock In {\em Proc.\ 40th Internat.\ Conf.\ Mach.\ Learn.}, 2023.
\newblock URL: \url{https://arxiv.org/abs/2302.09738}.

\bibitem{LYZ10}
Y.-S. Liu, J.~Yi, H.~Zhang, G.-Q. Zheng, and J.-C. Paul.
\newblock Surface area estimation of digitized 3d objects using quasi-monte carlo methods.
\newblock {\em Pattern Recognition}, 43:3900--3909, 2010.
\newblock \href {https://doi.org/10.1016/j.patcog.2010.06.002} {\path{doi:10.1016/j.patcog.2010.06.002}}.

\bibitem{NiS19}
F.~Nielsen and K.~Sun.
\newblock Clustering in {Hilbert's} projective geometry: The case studies of the probability simplex and the elliptope of correlation matrices.
\newblock In {\em Geometric Structures of Information}, Signals and communication technology, pages 297--331. Springer International Publishing, 2019.
\newblock \href {https://doi.org/10.1007/978-3-030-02520-5_11} {\path{doi:10.1007/978-3-030-02520-5_11}}.

\bibitem{PaT09}
A.~Papadopoulos and M.~Troyanov.
\newblock Harmonic symmetrization of convex sets and of {Finsler} structures, with applications to {Hilbert} geometry.
\newblock {\em Expo.\ Math.}, 27:109--124, 2009.
\newblock \href {https://doi.org/10.1016/j.exmath.2008.10.001} {\path{doi:10.1016/j.exmath.2008.10.001}}.

\bibitem{PaT14}
A.~Papadopoulos and M.~Troyanov.
\newblock {\em Handbook of {Hilbert} Geometry}.
\newblock EMS Press, 2014.
\newblock \href {https://doi.org/10.4171/147} {\path{doi:10.4171/147}}.

\bibitem{Roc97}
R.~T. Rockafellar.
\newblock {\em Convex Analysis}.
\newblock Princeton University Press, Princeton, NJ, 1997.
\newblock \href {https://doi.org/10.1515/9781400873173} {\path{doi:10.1515/9781400873173}}.

\bibitem{Rud87}
W.~Rudin.
\newblock {\em Real and Complex Analysis}.
\newblock McGraw Hill, 3rd edition, 1987.

\bibitem{Sch01}
R.~Schneider.
\newblock Crofton formulas in hypermetric projective {Finsler} spaces.
\newblock {\em Archiv der Mathematik}, 77:85--97, 2001.
\newblock \href {https://doi.org/10.1007/PL00000469} {\path{doi:10.1007/PL00000469}}.

\bibitem{Sch06a}
R.~Schneider.
\newblock Crofton measures in polytopal {Hilbert} geometries.
\newblock {\em Beitr.\ Algebra Geom.}, 47:479--488, 2006.
\newblock URL: \url{https://www.emis.de/journals/BAG/vol.47/no.2/12.html}.

\bibitem{Sch06b}
R.~Schneider.
\newblock Crofton measures in projective {Finsler} spaces.
\newblock In E.~L. Grinberg, S.~Li, G.~Zhang, and J.~Zhou, editors, {\em Integral Geometry and Convexity}, pages 67--98. World Scientific, 2006.
\newblock \href {https://doi.org/10.1142/9789812774644_0006} {\path{doi:10.1142/9789812774644_0006}}.

\bibitem{Sch14}
R.~Schneider.
\newblock {\em Convex Bodies: The {Brunn-Minkowski} Theory}, volume 151 of {\em Encyclopedia of Mathematics and its Applications}.
\newblock Cambridge Univ.\ Press, 2nd edition, 2014.
\newblock \href {https://doi.org/10.1017/CBO9781139003858} {\path{doi:10.1017/CBO9781139003858}}.

\bibitem{Tho96}
A.~C. Thompson.
\newblock {\em Minkowski Geometry}, volume~63 of {\em Encyclopedia of mathematics and its applications}.
\newblock Cambridge University Press, 1996.
\newblock \href {https://doi.org/10.1017/CBO9781107325845} {\path{doi:10.1017/CBO9781107325845}}.

\bibitem{Tro14}
M.~Troyanov.
\newblock Funk and {Hilbert} geometries from the {Finslerian} viewpoint.
\newblock In {\em Handbook of {Hilbert} geometry}, pages 69--110. European Mathematical Society Publishing House, 2014.
\newblock \href {https://doi.org/10.4171/147-1/3} {\path{doi:10.4171/147-1/3}}.

\bibitem{Zie95}
G{\"u}nter~M. Ziegler.
\newblock {\em Lectures on Polytopes}, volume 152 of {\em Graduate Texts in Mathematics}.
\newblock Springer-Verlag, New York, 1995.
\newblock \href {https://doi.org/10.1007/978-1-4613-8431-1} {\path{doi:10.1007/978-1-4613-8431-1}}.

\end{thebibliography}
\end{document}